\documentclass[12pt]{article}
\usepackage[margin=1in]{geometry}
\usepackage{amsmath}
\usepackage{amssymb}
\usepackage{amsthm}
\usepackage{xspace}
\usepackage[normalem]{ulem}
\usepackage{graphicx}
\usepackage{epsfig}
\usepackage{ifpdf}
\usepackage{url,hyperref}
\usepackage{latexsym}
\usepackage[mathscr]{euscript}
\usepackage{xspace}
\usepackage{color}
\usepackage{makeidx}
\usepackage{wrapfig,color}
\usepackage{stackrel}
\usepackage{algorithm}

\usepackage[all]{xy}
\usepackage{framed}
\usepackage{pb-diagram}

\long\def\remove#1{}

\usepackage{algpseudocode}
\newtheorem{theorem}{Theorem} % section

\newtheorem{corollary}[theorem]{Corollary}
\newtheorem{fact}[theorem]{Fact}
\newtheorem{observation}[theorem]{Observation}
\newtheorem{proposition}[theorem]{Proposition}
\newtheorem{definition}[theorem]{Definition}

\newtheorem{claim}{Claim}[section]

\newtheorem{remark}{Remark}[section]

\newcommand\id{\mathrm{id}}

\newcommand\R{\mathbb{R}}
\newcommand\Z{\mathbb{Z}}

\newcommand\diam{\mathrm{diam}}

\newcommand\MM{\mathrm{MM}}
\newcommand\M{\mathrm{M}}

\newcommand\res{\mathrm{res}}

\newcommand{\leb}{\lambda}

\definecolor{darkblue}{rgb}{0.0, 0.0, 0.8}
\definecolor{darkred}{rgb}{0.8, 0.0, 0.0}
\definecolor{darkgreen}{rgb}{0.0, 0.8, 0.0}

\newcommand{\eps}               {{\varepsilon}}
\newcommand{\denselist}{\itemsep 0pt\parsep=1pt\partopsep 0pt}

\newcommand{\NV}			{P_\delta}	
\newcommand{\myc}			{5}
\newcommand{\myCF}		{{\mathfrak C}}
\newcommand{\myDg}		{{\mathrm{Dg}}}
\newcommand{\myparagraph}[1]	{{\vspace*{0.06in}\noindent{\bf #1~}}}

\newif\ifpaper
\papertrue

\title{Topological Analysis of Nerves, Reeb Spaces, Mappers, and Multiscale Mappers}

\author{Tamal K. Dey\thanks{Department of Computer Science and Engineering, The Ohio State University. \texttt{tamaldey, yusu@cse.ohio-state.edu}}, Facundo M\'emoli\thanks{Department of Mathematics and Department of Computer Science and Engineering, The Ohio State University. \texttt{memoli@math.osu.edu}}, Yusu Wang$^*$}

\date{}

\begin{document}

\maketitle

%\newpage
\setcounter{page}{1}
%\linenumbers
\begin{abstract}
	Data analysis often concerns not only the space where data come from, but also various types of maps attached to data. In recent years, several related structures have been used to study maps on data, including Reeb spaces, mappers and multiscale mappers. The construction of these structures also relies on the so-called \emph{nerve} of a cover of the domain.
	
	In this paper, we aim to analyze the topological information encoded in these structures in order to provide better understanding of these structures and facilitate their practical usage.

	More specifically, we show that the one-dimensional homology of the nerve complex $N(\mathcal{U})$ of a path-connected cover $\mathcal{U}$ of a domain $X$ cannot be richer than that of the domain $X$ itself. Intuitively, this result means that no new $H_1$-homology class can be ``created'' under a natural map from $X$ to the nerve complex $N(\mathcal{U})$. Equipping $X$ with a pseudometric $d$, we further refine this result and characterize the classes of $H_1(X)$ that may survive in the nerve complex using the notion of \emph{size} of the covering elements in $\mathcal{U}$. These fundamental results about nerve complexes then lead to an analysis of the $H_1$-homology of Reeb spaces, mappers and multiscale mappers.
	
	The analysis of $H_1$-homology groups unfortunately does not extend to higher dimensions. Nevertheless, by using a map-induced metric, establishing a Gromov-Hausdorff convergence result between mappers and the domain, and interleaving relevant modules, we can still analyze the persistent homology groups of (multiscale) mappers to establish a connection to Reeb spaces. 
\end{abstract}

\section{Introduction}
\label{sec:intro}

Data analysis often concerns not only the space where data come from, but also various types of information attached to data. 
For example, each node in a road network can contain information about the average traffic flow passing this point, a node in protein-protein interaction network can be associated with biochemical properties of the proteins involved. 
%Such information attached to data can be modeled as maps defined on the domain of interest; a simple form being a scalar field defined on a space. 
Such information attached to data can be modeled as maps defined on the domain of interest; note that the maps are not necessarily $\mathbb{R}^d$-valued, e.g, the co-domain can be $\mathbb{S}^1$. 
%Hence understanding data needs not only methods to analyze spaces, but also maps defined on them. 
Hence understanding data benefits from analyzing maps relating two spaces rather than a single space with no map on it. 

In recent years, several related structures have been used to study general maps on data, including Reeb spaces \cite{munch,DW13,reeb-space,MW16}, mappers (and variants) \cite{CO16,CS14,mapper} and multiscale mappers \cite{DMW16}. 
%The construction of these structures also relies on the so-called \emph{nerve} of a cover of the domain. 
More specifically, given a map $f: X \to Z$ defined on a topological space $X$, the Reeb space $R_f$ w.r.t. $f$ (first studied for piecewise-linear maps in \cite{reeb-space}), is a generalization of the so-called Reeb graph for a scalar function which has been used in various applications \cite{BGSF08}. 
It is the quotient space of $X$ w.r.t. an equivalence relation that asserts two points of $X$ to be equivalent if they have the same function value and are connected to each other via points of the same function value. All equivalent points are collapsed into a single point in the Reeb space.
Hence $R_f$ provides a way to view $X$ from the perspective of $f$. 

The Mapper structure, originally introduced in \cite{mapper}, can be considered as a further  generalization of the Reeb space. 
Given a map $f: X \to Z$, it also considers a cover $\mathcal{U}$ of the co-domain $Z$ that enables viewing the structure of $f$ at a coarser level.
Intuitively, the equivalence relation between points in $X$ is now defined by whether points are within the same connected component of the pre-image of a cover element $U\in \mathcal{U}$. 
Instead of a quotient space, the mapper takes the nerve complex of the cover of $X$ formed by the connected components of the pre-images of all elements in $\mathcal{U}$ (i.e, the cover formed by those equivalent points). 
Hence the mapper structure provides a view of $X$ from the perspective of both $f$ and a cover of the co-domain $Z$. 

Finally, both the Reeb space and the mapper structures provide a fixed snapshot of the input map $f$. As we vary the cover $\mathcal{U}$ of the co-domain $Z$, we obtain a family of snapshots at different granularities. The \emph{multiscale mapper} \cite{DMW16} describes the sequence of the mapper structures as one varies the granularity of the cover of $Z$ through a sequence of covers of $Z$ connected via cover maps. 

%\paragraph*{New work.} 
\myparagraph{New work.} 
While these structures are meaningful in that they summarize the information contained in data, there has not been any qualitative analysis of the precise information encoded by them with the only exception of \cite{CO16} and \cite{GGP16} \footnote{Carri\`{e}re and Oudot~\cite{CO16} analyzed certain persistence diagram of mappers induced by a real-valued function, and provided a characterization for it in terms of the persistence diagram of the corresponding Reeb graph. 
	Gasparovic et al \cite{GGP16} provides full description of the persistence homology information encoded in the \emph{intrinsic \v{C}ech complex} (a special type of nerve complex) of a metric graph. }. 
In this paper, we aim to analyze the \emph{topological information} encoded by these structures, so as to provide better understanding of these structures and facilitate their practical usage~\cite{EH09,survey}.
In particular, the construction of the mapper and multiscale mapper use the so-called \emph{nerve} of a cover of the domain. %, which is a general topological construction. 
To understand the mappers and multiscale mappers, we first provide a quantitative analysis of the topological information encoded in the nerve of a reasonably well-behaved cover for a domain. Given the generality and importance of the nerve complex in topological studies, this result is of independent interest. 

More specifically, in Section \ref{sec:H1}, we first obtain a general result that relates the one dimensional homology $H_1$ of the nerve complex $N(\mathcal{U})$ of a path-connected cover $\mathcal{U}$ (where each open set contained is path-connected) of a domain $X$ to that of the domain $X$ itself. 
Intuitively, this result says that no new $H_1$-homology classes can be ``created" under a natural map from $X$ to the nerve complex $N(\mathcal{U})$. 
Equipping $X$ with a pseudometric $d$, we further refine this result and quantify the classes of $H_1(X)$ that may survive in the nerve complex (Theorem~\ref{H1prop-mapper}, Section \ref{sec:persistentH1}). This demarcation is obtained via a notion of \emph{size} of covering elements in $\mathcal{U}$. 
These fundamental results about nerve complexes then lead to an analysis of the $H_1$-homology classes in Reeb spaces (Theorem~\ref{RS-thm}), mappers and multiscale mappers (Theorem~\ref{H1pers-thm}). 
The analysis of $H_1$-homology groups unfortunately does not extend to higher dimensions. Nevertheless, we can still provide an interesting analysis of the persistent homology groups for these structures (Theorem~\ref{thm:MM-ICinterleave}, Section \ref{sec:highD}). During this course, by using a map-induced metric, we establish a Gromov-Hausdorff convergence between the mapper structure and the domain. This offers an alternative to \cite{MW16} for defining the convergence between mappers and the Reeb space, which may be of independent interest. 

\ifpaper {}\else
{\em All missing proofs in what follows are deferred to the full version of this paper on arXiv.}
\fi
%\myparagraph{Remark.} Our results on the topological characterization in some sense generalize the work of Gasparovic et al \cite{GGP16}, which provides a precise characterization of the persistence homology information encoded in the intrinsic \v{C}ech complex (a special type of nerve complex) of a metric graph. On one hand, our results work for any nerve complex over any domain. On the other hand, a tighter characterization of ``size'' of 1-dimensional homological features is obtained in \cite{GGP16}, given that it focuses on the intrinsic \v{C}ech nerve complex of a \emph{metric graph}. 

\section{Topological background and motivation}\label{sec:background}

%In this section we recall several facts about topological spaces, their covers, and associated nerve complexes.

%\paragraph*{Space, paths, covers.} 
\myparagraph{Space, paths, covers.} 
Let $X$ denote a path connected topological space. Since $X$ is path connected, there exists a path $\gamma:[0,1]\rightarrow X$ connecting every pair of points $\{x,x'\}\in X\times X$ where $\gamma(0)=x$ and $\gamma(1)=x'$. Let $\Gamma_X(x,x')$ denote the set of all such paths connecting $x$ and $x'$.
These paths play an important role in our definitions and
arguments.

By a cover of $X$ we mean a collection $\mathcal{U}=\{U_\alpha\}_{\alpha\in A}$ of open sets such that $\bigcup_{\alpha\in A} U_\alpha= X.$ A cover $\mathcal {U}$ is {\em path connected} if each $U_\alpha$ is path connected. In this paper, we consider only path connected covers.

%We need and hencenforth assume $X$ to be {\em paracompact}, which means that every open cover
%${\mathcal U}$ of $X$ has a subcover ${\mathcal U}'\subseteq {\mathcal U}$ so that each point $x\in X$ has an open neighborhood contained in {\em finitely many} elements of ${\mathcal U}'$. Such a cover ${\mathcal U}'$ is called {\em locally finite}. Paracompactness of $X$ is used to define maps between $X$ and its nerve complexes. From now on, we assume {\em $X$ to be compact} which implies that it is paracompact as well.
Later to define maps between $X$ and its nerve complexes, we need $X$ to be {\em paracompact}, that is, every cover
${\mathcal U}$ of $X$ has a subcover ${\mathcal U}'\subseteq {\mathcal U}$ so that each point $x\in X$ has an open neighborhood contained in {\em finitely many} elements of ${\mathcal U}'$. Such a cover ${\mathcal U}'$ is called {\em locally finite}. 
From now on, we assume {\em $X$ to be compact} which implies that it is paracompact too.
% as well.

\begin{definition} [Simplicial complex and maps] A simplicial complex $K$ with a vertex set $V$ is a collection of subsets of $V$ with the condition that if $\sigma \in 2^V$ is in $K$, then all subsets of $\sigma$ are in $K$. We denote the geometric realization of $K$ by $|K|$. Let $K$ and $L$ be two simplicial complexes. A map $\phi:K\rightarrow L$ is \emph{simplicial} if for every simplex $\sigma=\{v_1,v_2,\ldots, v_p\}$ in $K$, the simplex $\phi(\sigma)=\{\phi(v_1),\phi(v_2),\ldots,\phi(v_p)\}$ is in $L$. 
\end{definition}

%Given a simplicial complex by  $\id$ we will denote the \emph{identity simplicial map}, that is, the simplicial map arising from the identity map on vertices.

%\begin{definition}[Contiguous simplicial maps]
%Let $K$ and $L$ be two finite simplicial complexes and $h_1,h_2:K\rightarrow L$ be simplicial maps. If for all $\sigma\in K$ it holds that $h_1(\sigma)\cup h_2(\sigma)\in L$ then we say that $h_1$ and $h_2$ are \emph{contiguous}.
%\end{definition}
%Recall that contiguous simplicial maps induce the same map at homology level, cf. \cite{munkres}.

\begin{definition}[Nerve of a cover]
	Given a cover ${\mathcal U} = \{
	U_{\alpha}\}_{\alpha \in A}$ of $X$, we define the {\em
		nerve} of the cover ${\mathcal U}$ to be the simplicial complex
	$N({\mathcal U})$ whose vertex set is the index set $A$, and where
	a subset $\{ \alpha _0 , \alpha _1, \ldots , \alpha _k \}\subseteq A$ spans a
	$k$-simplex in $N({\mathcal U})$ if and only if $U_{\alpha _0 } \cap
	U_{\alpha _1 } \cap \ldots \cap U_{\alpha _k} \neq \emptyset$.  
\end{definition}

%\subsection{Maps between covers}
%\paragraph*{Maps between covers. } 
\myparagraph{Maps between covers.} 
Given two covers ${\mathcal U} = \{ U_{\alpha} \} _{\alpha \in A}$ and ${\mathcal V} = \{ V_{\beta} \} _{\beta \in B}$ of $X$, a {\em map of covers} from ${\mathcal U} $ to ${\mathcal V}$ is a set map $\xi: A \rightarrow B$ so that $U_{\alpha} \subseteq V_{\xi(\alpha)}$ for all $\alpha \in A$. {\em By a slight abuse of notation we also use $\xi$ to indicate the map $\mathcal{U}\rightarrow\mathcal{V}.$} Given such a map of covers, there is an induced simplicial map $N(\xi) : N({\mathcal U}) \rightarrow N({\mathcal V})$, given on vertices by the map $\xi$.  
Furthermore, if $\mathcal{U}\stackrel{\xi}{\rightarrow}\mathcal{V}\stackrel{\zeta}{\rightarrow}\mathcal{W}$ are three covers of $X$ with the intervening maps of covers between them, then $N(\zeta\circ\xi) = N(\zeta)\circ N(\xi)$ as well.
The following simple result is useful.
\begin{proposition}[Maps of covers induce contiguous simplicial maps \cite{DMW16}]\label{prop:cover-contiguity}
	Let $\zeta,\xi:\mathcal{U}\rightarrow\mathcal{V}$ be any two maps of covers. Then, the simplicial maps $N(\zeta)$ and $N(\xi)$ are contiguous.
\end{proposition}

Recall that two simplicial maps $h_1, h_2: K \rightarrow L$ are \emph{contiguous} if for all $\sigma\in K$ it holds that $h_1(\sigma)\cup h_2(\sigma)\in L$. In particular, contiguous maps induce identical maps at the homology level \cite{munkres}. Let $H_k(\cdot)$ denote the $k$-dimensional homology of the space in its argument. This homology is
{\em singular} or {\em simplicial} depending on if the argument is a topological space or a simplicial complex respectively. All homology groups
in this paper are defined over the field $\Z_2$.
Proposition \ref{prop:cover-contiguity} implies that the map $H_k(N(\mathcal{U}))\rightarrow H_k(N(\mathcal{V}))$ arising out of a cover map can be deemed canonical.

\section{Surjectivity in $H_1$-persistence}
\label{sec:H1}

In this section we first establish a map $\phi_{\mathcal U}$ between $X$ and the geometric realization $|N(\mathcal U)|$ of a nerve complex $N(\mathcal U)$. 
%This helps us to define a map $\phi_{\mathcal U}*$ from the singular homology groups of $X$ to the singular homology groups of $|N(\mathcal U)|$ and hence, to the simplicial homology groups  of $N(\mathcal U)$. 
This helps us to define a map  $\phi_{\mathcal U}*$ from the singular homology groups of $X$ to the simplicial homology groups 
of $N(\mathcal U)$ (through the singular homology of $|N(\mathcal U)|$). 
The famous nerve theorem~\cite{borsuk,leray} says that if the elements of $\mathcal U$ intersect only in contractible spaces, then $\phi_{\mathcal U}$ is a homotopy equivalence and hence $\phi_{\mathcal U}*$ leads to an isomorphism between $H_\ast(X)$ and $H_\ast(N(\mathcal U))$. The contractibility condition can be weakened to a {\em homology ball} condition to retain the isomorphism between the two homology groups~\cite{leray}. In absence of such conditions of the cover, simple examples
exist to show that $\phi_{\mathcal U}*$ is neither a monophorphism (injection) nor an epimorphism (surjection). Figure~\ref{non-surject-fig} gives an example where $\phi_{{\mathcal U}*}$ is not sujective in $H_2$. However, for one dimensional homology we show that, for any path connected cover $\mathcal U$, the
map $\phi_{{\mathcal U}*}$ is necessarily a surjection. One implication of this is that the simplicial maps arising out of cover maps induce a surjection among the one dimensional homology groups of two nerve complexes.

%\begin{figure*}[t]
%	\begin{center}
%		\includegraphics[width=0.8\textwidth]{non-surject-fig.pdf}
%	\end{center}
%\vspace*{-0.1in}
%	\caption{The map $f:\mathbb{S}^2\subset \mathbb{R}^3\rightarrow \mathbb{R}^2$ takes the sphere to $\mathbb{R}^2$. The pullback of the cover element $U_\alpha$ makes a band surrounding the equator which causes the nerve $N(f^{-1}{\mathcal U})$ to pinch in the middle creating two $2$-cycles. This shows that the map $\phi_*: X\rightarrow N(*)$ may not induce a surjection in $H_2$.}
%	\label{non-surject-fig}
%\vspace*{-0.15in}
%\end{figure*}

\subsection{Nerves}

\begin{wrapfigure}{r}{0.2\textwidth}
	\vspace*{-0.4in}
	\begin{align*}
	\xymatrix{ & X_{\mathcal{U}}\ar[rd]^{\pi}  & \\
		X \ar[ru]^{\zeta}  \ar[rr]^{\phi_{\mathcal U}} & & |N(\mathcal{U})| 
	}
	\end{align*}
	\vspace*{-0.3in}
\end{wrapfigure} 
The proof of the nerve theorem~\cite{hatcher} uses a construction that connects the two spaces $X$ and $|N(\mathcal U)|$ via a third space $X_{\mathcal U}$ that is a product space of $\mathcal U$ and the geometric realization $|N({\mathcal U})|$. In our case $\mathcal U$ may not satisfy the contractibility condition. Nevertheless, we use the same construction to define  three maps, $\zeta: X \rightarrow X_{\mathcal U}$, 
$\pi : X_{\mathcal U}\rightarrow |N({\mathcal U})|$, and $\phi_{\mathcal U} : X\rightarrow |N({\mathcal U})|$ where $\phi_{\mathcal U}= \pi\circ \zeta$ is referred to as the {\em nerve map}. Details about the construction of these maps follow. 
\begin{wrapfigure}{l}{0.25\textwidth}
	%	\vspace*{-0.2in}
	\includegraphics[width=0.25\textwidth]{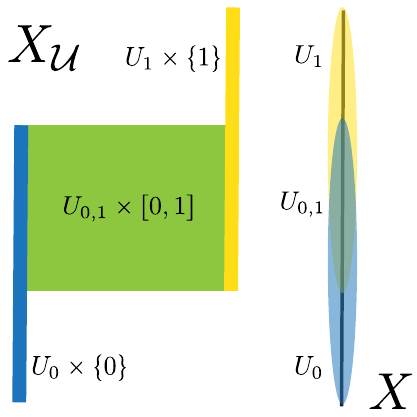}
\end{wrapfigure} 

Denote the elements of the cover $\mathcal{U}$ as $U_\alpha$ for $\alpha$ taken from some indexing set $A$. The vertices of $N(\mathcal{U})$ are denoted by $\{u_\alpha,\,\alpha\in A\},$ where each $u_\alpha$ corresponds to the cover element $U_\alpha.$ For each finite non-empty intersection
$U_{\alpha_0,\ldots,\alpha_n}:=\bigcap_{i=0}^n U_{\alpha_i}$ consider the product $U_{\alpha_0,\ldots,\alpha_n}\times \Delta^n_{\alpha_0,\ldots,\alpha_n}$, where $\Delta^n_{\alpha_0,\ldots,\alpha_n}$ denotes the $n$-dimensional simplex with vertices $u_{\alpha_0},\ldots,u_{\alpha_n}$. Consider now the disjoint union

$$M:=\bigsqcup_{\alpha_0,\ldots,\alpha_n\in A:\,U_{\alpha_0,\ldots,\alpha_n}\neq \emptyset} U_{\alpha_0,\ldots,\alpha_n}\times \Delta_{\alpha_0,\ldots,\alpha_n}^n$$
together with the following identification: each point $(x,y)\in M$, with $x\in U_{\alpha_0,\ldots,\alpha_n}$ and $y\in [\alpha_0,\ldots,\widehat{\alpha}_i,\ldots,\alpha_n]\subset \Delta_{\alpha_0,\ldots,\alpha_n}^n$ is identified with the corresponding point in the product $U_{\alpha_0,\ldots,\widehat{\alpha}_i,\ldots,\alpha_n}\times \Delta_{\alpha_0,\ldots,\widehat{\alpha}_i,\ldots,\alpha_n}$ via the inclusion $U_{\alpha_0,\ldots,\alpha_n}\subset U_{\alpha_0,\ldots,\widehat{\alpha}_i,\ldots,\alpha_n}$. Here $[\alpha_0,\ldots,\widehat{\alpha}_i,\ldots,\alpha_n]$ denotes the $i$-th face of the simplex $\Delta_{\alpha_0,\ldots,\alpha_n}^n.$ Denote by $\sim$ this identification and now define the space
$X_\mathcal{U}:= M~/ \sim.$
An example for the case when $X$ is a line segment and $\mathcal{U}$ consists of only two open sets is shown in the previous page.

%\begin{wrapfigure}{r}{0.25\textwidth}
%\vspace*{-0.2in}
%		\includegraphics[width=0.25\textwidth]{blowup-vert.pdf}
%\end{wrapfigure} 

%\facundo{I changed the below:}

\begin{definition}
	A collection of real valued continuous functions $\{\varphi_\alpha:\rightarrow[0,1],\alpha\in A\}$ is called
	a {\em partition of unity } if (i)
	$\sum_{\alpha\in A}\varphi_\alpha(x)=1$ for all $x\in X$,
	(ii) For every $x\in X$, there are only finitely many $\alpha\in A$ such that $\varphi_\alpha(x)>0.$
	
	If $\mathcal{U} = \{U_\alpha,\,\alpha\in A\}$ is any open cover of $X$, then a partition of unity $\{\varphi_\alpha,\,\alpha\in A\}$ is \emph{subordinate} to $\mathcal{U}$ if $\mathrm{supp}(\varphi_\alpha)$ is contained in $U_\alpha$ for each $\alpha\in A$. 
\end{definition}

\begin{figure*}[ht]
	\begin{center}
		\includegraphics[width=0.8\textwidth]{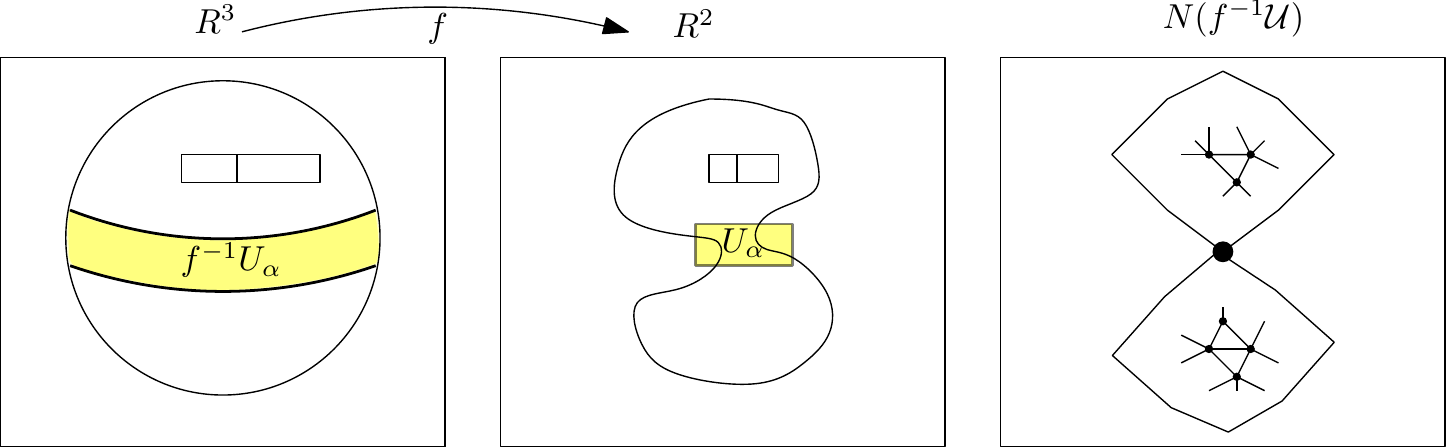}
	\end{center}
	\vspace*{-0.1in}
	\caption{The map $f:\mathbb{S}^2\subset \mathbb{R}^3\rightarrow \mathbb{R}^2$ takes the sphere to $\mathbb{R}^2$. The pullback of the cover element $U_\alpha$ makes a band surrounding the equator which causes the nerve $N(f^{-1}{\mathcal U})$ to pinch in the middle creating two $2$-cycles. This shows that the map $\phi_*: X\rightarrow N(*)$ may not induce a surjection in $H_2$.}
	\label{non-surject-fig}
%	\vspace*{-0.15in}
\end{figure*}

% {subordinate to} $\mathcal U$ if the support 
% 	$\mathrm{supp}(\varphi_\alpha)$ is contained in $U_\alpha$ for each $\alpha\in A$ and
% 	$\Sigma_{\alpha\in A} \varphi_\alpha(x)=1$ for each $x\in X$.

%and it is called {\em locally finite} if for each $x\in X$,
%$\varphi_\alpha(x)\not= 0$ only for a finitely many $\alpha\in A$. Since $X$ is %paracompact, there exists a locally finite partition of unity
%subordinate to any open cover.

Since $X$ is paracompact, for any open cover $\mathcal{U}=\{U_\alpha,\,\alpha\in A\}$ of $X$, there exists a partition of unity $\{\varphi_\alpha,\alpha\in A\}$ subordinate to $\mathcal U$ \cite{prasolov}. For each $x\in X$ such that $x\in U_\alpha$, denote by $x_\alpha$ the corresponding copy of $x$ residing in $X_\mathcal{U}$.  Then, the map $\zeta:X\rightarrow X_\mathcal{U}$ is defined as follows: for any $x\in X$, 
$$ \zeta(x):=\sum_{\alpha\in A}\varphi_\alpha(x)\,x_\alpha.$$

The map $\pi:X_\mathcal{U}\rightarrow |N(\mathcal{U})|$ is induced by the individual projection maps  $$U_{\alpha_0,\ldots,\alpha_n}\times \Delta^n_{\alpha_0,\ldots,\alpha_n}\rightarrow  \Delta^n_{\alpha_0,\ldots,\alpha_n}.$$

Then,  it follows that $\phi_{\mathcal U}=\pi\circ\zeta:X\rightarrow |N(\mathcal{U})|$  satisfies, for $x\in X$,
\begin{align}\label{eq:phiU}
\phi_{\mathcal U}(x)&=\sum_{\alpha\in A}\varphi_\alpha(x)\,u_\alpha.
\end{align}

We have the following fact~\cite[pp. 108]{prasolov}:
\begin{fact}
	$\zeta$ is a homotopy equivalence. 
	\label{homotopy-fact}
\end{fact}

\subsection{From space to nerves}
Now, we show that the nerve maps at the homology level are surjective for one dimensional homology when the covers are path-connected. Interestingly, the result is not true beyond one dimensional homology (see Figure~\ref{non-surject-fig}) which is probably why this simple but important fact has not been observed before. First, we make a simple observation that connects the classes in singular homology of $|N(\mathcal U)|$ to those in the simplicial homology of $N(\mathcal U)$. The result follows immediately from the isomorphism between singular and simplicial homology induced by the geometric realization; see~\cite[Theorem 34.3]{munkres}. In what follows let $[c]$ denote the class of a cycle $c$.

\begin{proposition}
	Every $1$-cycle $\xi$ in $|N({\mathcal U})|$ has a $1$-cycle $\gamma$ in $N({\mathcal U})$ so that $[\xi]=[|\gamma|]$.
	\label{embed-prop}
\end{proposition}

\begin{proposition}
	If $\mathcal U$ is path connected,	$\phi_{{\mathcal U}*}: H_1(X)\rightarrow H_1(|N({\mathcal U})|)$ is a surjection.  \label{phi-surject-prop}
\end{proposition}
\begin{proof}
	Let $[\gamma]$ be any class in $H_1(|N({\mathcal U})|)$. Because of Proposition~\ref{embed-prop}, we can assume that $\gamma = |\gamma'|$, where $\gamma'$ is a $1$-cycle in the $1$-skeleton of $N({\mathcal U})$. We construct a $1$-cycle $\gamma_{\mathcal U}$ in $X_{\mathcal U}$ so that $\pi(\gamma_{\mathcal U})=\gamma$. Recall the map $\zeta: X\rightarrow X_{\mathcal U}$ in the construction of the nerve map
	$\phi_{\mathcal U}$ where
	$\phi_{\mathcal U}=\pi \circ \zeta$. There exists a class $[\gamma_X]$ in $H_1(X)$ so that $\zeta_*([\gamma_X])=[\gamma_{\mathcal U}]$ because $\zeta_*$ is an isomorphism by Fact~\ref{homotopy-fact}. Then, $\phi_{{\mathcal U}*}([\gamma_X])=\pi_*(\zeta_*([\gamma_X]))$ because $\phi_{{\mathcal U}*}=\pi_*\circ\zeta_*$. It follows $\phi_{{\mathcal U}*}([\gamma_X])=\pi_*([\gamma_{\mathcal U}])=[\gamma]$ showing that $\phi_{{\mathcal U}*}$ is surjective.
	
	\ifpaper
	Therefore, it remains only to show that a $1$-cycle $\gamma_{\mathcal U}$ can be constructed given $\gamma'$ in $N({\mathcal U})$ so that $\pi(\gamma_{\mathcal U})=\gamma=|\gamma'|$. Let $e_0,e_1,\ldots,e_{r-1},e_r=e_0$ be an ordered sequence of edges on $\gamma$. Recall the construction of the space $X_{\mathcal U}$. In that terminology, let $e_i=\Delta^n_{\alpha_i\alpha_{(i+1)\mod r}}$. Let $v_i=e_{(i-1)\mod r}\cap e_{i}$ for $i\in [0,r-1]$.
	The vertex $v_i=v_{\alpha_i}$ corresponds to the cover element $U_{\alpha_i}$ where
	$U_{\alpha_i}\cap U_{\alpha_{(i+1)\mod r}}\not=\emptyset$ for every $i\in [0,r-1]$.
	Choose a point $x_i$ in the common intersection $U_{\alpha_i}\cap U_{\alpha_{(i+1)\mod r}}$ for every $i\in [0,r-1]$.  Then, the {\em edge} path $\tilde{e_i}=e_i\times x_{i}$ is in $X_{\mathcal U}$ by construction. Also, letting $x_{\alpha_i}$ to be the lift of $x_i$ in the lifted $U_{\alpha_i}$, we can choose a {\em vertex} path $x_{\alpha_i}\leadsto x_{\alpha_{(i+1)\mod r}}$ residing in the lifted $U_{\alpha_i}$ and hence in $X_{\mathcal U}$ because $U_{\alpha_i}$ is path connected. Consider the following cycle obtained by concatenating the edge and vertex paths
	$$
	\gamma_{\mathcal U}=\tilde{e}_0x_{\alpha_0}\leadsto x_{\alpha_1}\tilde{e}_1\cdots \tilde{e}_{r-1}x_{\alpha_{r-1}}\leadsto x_{\alpha_0}
	$$
	By projection, we have $\pi(\tilde e_i)=e_i$ for every $i\in [0,r-1]$ and
	$\pi(x_{\alpha_i}\leadsto x_{\alpha_{(i+1)\mod r}})=v_{\alpha_i}$ and thus
	$\pi(\gamma_{\mathcal U})=\gamma$ as required.
	\else
	Therefore, it remains only to show that a $1$-cycle $\gamma_{\mathcal U}$ can be constructed given $\gamma$ in $N({\mathcal U})$ so that $\pi(\gamma_{\mathcal U})=\gamma$. See the full version for this construction.
	\fi
\end{proof}
Since we are eventually interested in the simplicial homology groups of the nerves rather than the singular homology groups of their geometric realizations, we make one more transition using the known isomorphism between the two homology groups. Specifically, if $\iota_{\mathcal U}: H_k(|N({\mathcal U})|)\rightarrow H_k(N({\mathcal U}))$ denotes this isomorphism, we let $\bar{\phi}_{{\mathcal U}*}$ denote the composition $\iota_{\mathcal U}\circ\phi_{{\mathcal U}*}$. As a corollary to Proposition~\ref{phi-surject-prop}, we obtain:
\begin{theorem}
	If $\mathcal U$ is path connected, $\bar{\phi}_{{\mathcal U}*}: H_1(X)\rightarrow H_1(N({\mathcal U}))$ is a surjection.
	\label{surj-map-thm}
\end{theorem}

\subsection{From nerves to nerves}
In this section we extend the result in Theorem~\ref{surj-map-thm} to simplicial maps between two nerves induced by cover maps. The following proposition is key to establishing the result.

\begin{proposition}[Coherent partitions of unity]\label{prop:coherent}
	Suppose $\{U_\alpha\}_{\alpha\in A}=\mathcal{U}\stackrel {\theta}{\longrightarrow}\mathcal{V}=\{V_\beta\}_{\beta\in B}$ are open covers of the paracompact topological space $X$ and $\theta: A\rightarrow B$ is a map of covers. Then there exists a partition of unity $\{\varphi_\alpha\}_{\alpha\in A}$ subordinate to the cover $\mathcal{U}$ such that if for each $\beta\in B$ we define 
	$$\psi_\beta:=\left\{
	\begin{array}{ll}
	\sum_{\alpha\in\theta^{-1}(\beta)}\varphi_\alpha & \mbox{if $\beta\in\mathrm{im}(\theta)$;}\\
	0&\mbox{otherwise}.
	\end{array}\right.$$
	then the set of functions $\{\psi_\beta\}_{\beta\in B}$ is a partition of unity subordinate to the cover $\mathcal{V}.$
\end{proposition}
\ifpaper
\begin{proof}
	The proof closely follows that of \cite[Corollary pp. 97]{prasolov}. Since $X$ is paracompact, there exists a locally finite refinement $\mathcal{W}=\{W_\lambda\}_{\lambda\in  L}$ of $\mathcal{U}$, a refinement map $L\stackrel{\xi}{\rightarrow} A$,  and a partition of unity $\{\omega_\lambda\}_{\lambda\in L}$ subordinate to $\mathcal{W}.$ For each $\alpha\in A$ define 
	$$\varphi_\alpha := 
	\left\{\begin{array}{ll}
	\sum_{\lambda\in\xi^{-1}(\alpha)}\omega_\lambda & \mbox{if $\alpha\in\mathrm{im}(\xi)$;}\\
	0&\mbox{otherwise}.
	\end{array}\right.$$
	The fact that the sum is well defined and continuous follows from the fact that $\mathcal{W}$ is locally finite. Let $C_\alpha:=\bigcup_{\lambda\in\xi^{-1}(\alpha)}\mathrm{supp}(\omega_\lambda)$. The set $C_\alpha$ is closed, $C_\alpha\subset U_\alpha$, and $\varphi_\alpha(x) = 0$ for $x\notin C_\alpha$ so that $\mathrm{supp}(\varphi_\alpha)\subset C_\alpha\subset U_\alpha.$ Now, to check that the family $\{C_\alpha\}_{\alpha\in A}$ is locally finite pick any point $x\in X$. Since $\mathcal{W}$ is locally finite there is an open set $O$ containing $x$ such that $O$ intersects only finitely many elements in $\mathcal{W}$. Denote these cover elements by $W_{\lambda_1},\ldots,W_{\lambda_N}.$ Now, notice if $\alpha\in A$  and $\alpha\notin\{\xi(\lambda_i),i=1,\ldots,N\}$, then $O$ does not intersect $C_\alpha$. Then, the family $\{\mathrm{supp}(\varphi_\alpha)\}_{\alpha\in A}$ is locally finite. It then follows that for $x\in X$ one has
	$$\sum_{\alpha\in A}\varphi_\alpha(x) = \sum_{\alpha\in A}\sum_{\lambda\in\xi^{-1}(\alpha)}\omega_\lambda(x)=\sum_{\lambda\in L}\omega_{\lambda}(x)=1.$$
	
	We have obtained that $\{\varphi_\alpha\}_{\alpha\in A}$ is a partition of unity subordinate to $\mathcal{U}$. Now, the same argument can be applied to the family $\{\psi_\beta\}_{\beta\in B}$ to obtain the proof of the proposition.
\end{proof}
\else
Proof is deferred to the full version.
\fi

Let $\{U_\alpha\}_{\alpha\in A}=\mathcal{U}\stackrel {\theta}{\longrightarrow}\mathcal{V}=\{V_\beta\}_{\beta\in B}$  be two open covers of $X$ connected by a map of covers. Apply Proposition \ref{prop:coherent} to obtain coherent partitions of unity $\{\varphi_\alpha\}_{\alpha\in A}$ and $\{\psi_\beta\}_{\beta\in B}$ subordinate to $\mathcal{U}$ and $\mathcal{V}$, respectively. Let the nerve maps $\phi_{{\mathcal U}}: X\rightarrow |N({\mathcal U})|$ and  $\phi_{{\mathcal V}}: X\rightarrow |N({\mathcal V})|$ be defined as in (\ref{eq:phiU}) above.	Let $N({\mathcal U})\stackrel{\tau}{\rightarrow} N({\mathcal V})$ be the simplicial map induced by the cover map $\theta$. Then, $\tau$ can be extended to a continuous map $\hat{\tau}$ on the image of $\phi_{\mathcal U}$ as follows: for $x\in X$,  $\hat{\tau}(\phi_{\mathcal U}(x))=\Sigma_{\alpha\in A}\varphi_{\alpha}(x)\,v_{\theta(\alpha)}$.

\begin{proposition}
	Let ${\mathcal U}$ and ${\mathcal V}$ be two covers of $X$ connected by a cover map ${\mathcal U}\stackrel{\theta}{\rightarrow} {\mathcal V}$. Then, the nerve maps $\phi_{\mathcal U}$ and
	$\phi_{\mathcal V}$ satisfy $\phi_{\mathcal V}=\hat{\tau}\circ \phi_{\mathcal U}$ where
	$\tau: N({\mathcal U})\rightarrow N({\mathcal V})$ is the simplicial map induced by the cover map $\theta$.
\end{proposition}
\begin{proof}
%\noindent{{\bf {Proof.}}} 
For any point  $p\in \mathrm{im}(\phi_{\mathcal{U}})$, there is $x\in X$ where $p=\phi_{\mathcal U}(x)=\Sigma_{\alpha\in A}\varphi_{\alpha}(x)u_{\alpha}$. Then,
\begin{align*}
\hat{\tau}\circ \phi_{\mathcal U}(x)&=\hat{\tau}\left(\sum_{\alpha\in A}\varphi_{\alpha}(x) u_\alpha\right)
= \sum_{\alpha\in A}\varphi_{\alpha}(x)\tau(u_\alpha)
= \sum_{\alpha\in A}\varphi_{\alpha}(x)\,v_{\theta(\alpha)}\\
&= \sum_{\beta\in B}\sum_{\alpha\in \theta^{-1}(\beta)}\varphi_{\alpha}(x)\,v_{\theta(\alpha)}
= \sum_{\beta\in B}\psi_\beta(x) v_\beta
=\phi_{\mathcal V}(x)
\end{align*}
\end{proof}

An immediate corollary of the above Proposition is:
\begin{corollary}
	The induced maps of $\phi_{{\mathcal U}*}: H_k(X)\rightarrow H_k(|N({\mathcal U})|)$, $\phi_{{\mathcal V}*}: H_k(X)\rightarrow H_k(|N({\mathcal V})|)$, and 
	$\hat{\tau}_*: H_k(|N({\mathcal U})|)\rightarrow H_k(|N({\mathcal V})|)$ at the homology levels commute, that is, $\phi_{{\mathcal V}*}=\hat{\tau}_*\circ \phi_{{\mathcal U}*}$.
	\label{cor-commute}
\end{corollary}

With transition from singular to simplicial homology, Corollary~\ref{cor-commute} implies that:

\begin{proposition}
	$\bar{\phi}_{{\mathcal V}*}=\tau_*\circ \bar{\phi}_{{\mathcal U}*}$ where 
	$\bar{\phi}_{{\mathcal V}*}: H_k(X)\rightarrow H_k(N({\mathcal V}))$, 
	$\bar{\phi}_{{\mathcal U}*}: H_k(X)\rightarrow H_k(N({\mathcal U}))$
	and $\tau: N({\mathcal U})\rightarrow N({\mathcal V})$ is the simplicial map induced by a cover map ${\mathcal U}\rightarrow {\mathcal V}$.
	\label{commute-prop}
\end{proposition}

Proposition~\ref{commute-prop} extends Theorem~\ref{surj-map-thm} to the simplicial maps between two nerves.

\begin{theorem}
	Let $\tau: N({\mathcal U})\rightarrow N({\mathcal V})$ be a simplicial map induced by a cover map ${\mathcal U}\rightarrow {\mathcal V}$ where both $\mathcal U$ and $\mathcal V$ are path connected. Then, 
	$\tau_*: H_1(N({\mathcal U}))\rightarrow H_1(N({\mathcal V}))$ is a surjection.
\end{theorem}
\begin{proof}
	Consider the maps 
	$$
	H_1(X)\stackrel{\bar{\phi}_{{\mathcal U}*}}{\rightarrow}H_1(N({\mathcal U}))\stackrel{\tau_*}{\rightarrow} H_1(N({\mathcal V})), \mbox{ and } H_1(X)\stackrel{\bar{\phi}_{{\mathcal V}*}}{\rightarrow} H_1(N({\mathcal V})).
	$$
	By Proposition~\ref{commute-prop}, $\tau_*\circ \bar{\phi}_{{\mathcal U}*}=\bar{\phi}_{{\mathcal V}*}$. By Theorem~\ref{surj-map-thm}, the map $\bar{\phi}_{{\mathcal V}*}$ is a surjection. It follows that $\tau_*$ is a surjection.
\end{proof}

\subsection{Mapper and multiscale mapper}
In this section we extend the previous results to the structures called
mapper and multiscale mapper. Recall that $X$ is assumed to be compact.
Consider a cover of $X$ obtained indirectly as a pullback
of a cover of another space $Z$. This gives rise to the so called
{\em Mapper} and {\em Multiscale Mapper}.
Let $f: X \rightarrow Z$ be a continuous map
where $Z$ is equipped with an open cover ${\mathcal U} = \{
U_{\alpha} \}_{\alpha \in A}$
%, again 
for some index set
$A$.  Since $f$ is continuous, the sets $\{f^{-1}(U_{\alpha}),\,\alpha\in A\}$ 
form an open cover of $X$.  For each $\alpha$, we can now consider
the decomposition of $f^{-1}(U_{\alpha })$ into its path connected
components, so we write $f^{-1}(U_{\alpha }) = \bigcup _{i =
	1}^{j_{\alpha}} V_{\alpha , i }$, where $j_{\alpha}$ is the number of
path connected components $V_{\alpha,i}$'s in $f^{-1}(U_{\alpha })$.  We write
$f^\ast{\mathcal U} $ for the cover of $X$ obtained this way
from the cover ${\mathcal U}$ of $Z$ and refer to it as the \emph{pullback} cover of $X$ induced by $\mathcal{U}$ via $f$. 
Note that by its construction, this pullback cover $f^\ast{\mathcal U} $ is path-connected. 

Notice that there are pathological examples of $f$
where $f^{-1}(U_\alpha)$ may shatter into infinitely many path components. This motivates us to
consider \emph{well-behaved} functions $f$: we require that for every 
path connected open set $U\subseteq Z$, the preimage $f^{-1}(U)$ has \emph{finitely} many open path connected components. Henceforth, all such functions are assumed to be well-behaved.

%\begin{center}
%\emph{If not stated otherwise, all functions are assumed to be well-behaved }
%\end{center}

\begin{definition}[Mapper \cite{mapper}]\label{def:mapper}
	Let $f:X\rightarrow Z$ be a continuous map. Let $\mathcal{U} = \{U_\alpha\}_{\alpha\in A}$ be an open cover of $Z$.
	The \emph{mapper} arising from these data is defined to be the nerve simplicial complex of the pullback cover:
	$\mathrm{M}(\mathcal{U},f) := {N}(f^\ast\mathcal{U}).$ 
\end{definition}

%\paragraph{Maps between pullbacks.} 
When we consider a continuous map $f: X \rightarrow Z$ and we are given a map of covers $\xi:{\mathcal
	U} \rightarrow {\mathcal V}$ between covers of $Z$, we observed in~\cite{DMW16} that there is a corresponding map of
covers between the respective pullback covers of $X$: 
$f^\ast(\xi):f^\ast{\mathcal U} \longrightarrow f^\ast{\mathcal
	V}.$ 
Furthermore, if $\mathcal{U}\stackrel{\xi}{\rightarrow}\mathcal{V}\stackrel{\theta}{\rightarrow}\mathcal{W}$ are three different covers of a topological space with the intervening maps of covers between them, then $f^\ast(\theta\circ\xi) = f^\ast(\theta)\circ f^\ast(\xi).$

%\subsection{Introduction to multiscale mapper and its persistent diagram}
In the definition below, \emph{objects} can be covers, simplicial complexes, or vector spaces.
\begin{definition}[Tower]\label{def:hfc}
	A \emph{tower} $\mathfrak{W}$ with resolution $r\in\R$ 
	is any collection 
	$\mathfrak{W}=\big\{\mathcal{W}_\eps\big\}_{\eps\geq r}$
	of objects $\mathcal{W}_\varepsilon$ indexed in $\R$ together with maps $w_{\varepsilon,\varepsilon'}: {\mathcal W} _\varepsilon \rightarrow {\mathcal W}_{\varepsilon'}$  
	so that $w_{\varepsilon,\varepsilon}=\id$ and 
	$w_{\varepsilon',\varepsilon''}\circ w_{\varepsilon,\varepsilon'} = w_{\varepsilon,\varepsilon''}$ for all $r\leq \varepsilon\leq \varepsilon'\leq \varepsilon''$.
	Sometimes we write
	$\mathfrak{W}=\big\{\mathcal{W}_\varepsilon\overset{\tiny{w_{\varepsilon,\varepsilon'}}}{\longrightarrow} \mathcal{W}_{\varepsilon'}\big\}_{r\leq\varepsilon\leq \varepsilon'}$ to denote the collection with the maps. 
	Given such a tower $\mathfrak{W}$, $\res(\mathfrak{W})$ refers to its resolution. 
	
	When $\mathfrak{W}$ is a collection of covers equipped with maps of covers between them, we call it \emph{a tower of covers}. 
	When $\mathfrak{W}$ is a collection of simplicial complexes equipped with simplicial maps between them, we call it \emph{a tower of simplicial complexes}. 
\end{definition}

The pullback properties described at the end of section~\ref{sec:background} make it possible to take the pullback of a given  tower of covers of a space via a given continuous function into another space, so that we obtain the following.

\begin{proposition}[\cite{DMW16}]\label{coro:pullback-cover}
	Let $\mathfrak{U}=\{{\mathcal U}_\eps\}$ be a tower of covers of $Z$ and $f:X\rightarrow Z$ be a continuous function.
	Then, $f^\ast\mathfrak{U}=\{f^\ast{\mathcal U}_\eps\}$ is a tower of (path-connected) covers of $X$.
\end{proposition}

In general, given a tower of covers $\mathfrak{W}$ of a space $X$, the nerve 
of each cover  in $\mathfrak{W}$ together with each map of $\mathfrak{W}$ provides a tower of simplicial complexes which we denote by $N(\mathfrak{W})$. 
\begin{definition}[Multiscale Mapper \cite{DMW16}]
	Let $f:X\rightarrow Z$ be a continuous map. Let $\mathfrak{U}$ be a tower of covers of $Z$.
	Then, the \emph{multiscale mapper} is defined to be the tower of the nerve simplicial complexes of the pullback:
	$\mathrm{MM}(\mathfrak{U},f):=N(f^\ast\mathfrak{U}).$
\end{definition}

As we indicated earlier, in general, no surjection between $X$ and its nerve may exist at the homology level. It follows that the same is true for the mapper $N(f^\ast{\mathcal U})$. But for $H_1$, we can apply the results contained in previous section to claim the following. 
\begin{theorem}
	Consider the following multiscale mapper arising out of a tower of path connected covers:
	\begin{equation*}
	N(f^\ast\mathcal{U}_{{0}})\rightarrow N(f^\ast\mathcal{U}_{{1}})\rightarrow\cdots\rightarrow N(f^\ast\mathcal{U}_{{n}})
	\end{equation*}
	\begin{itemize}
		\item There is a surjection from $H_1(X)$ to $H_1(N(f^*{\mathcal U}_{i}))$ for each $i\in [0,n]$.
		\item Consider a $H_1$-persistence module of a multiscale mapper as shown below.
		\begin{equation}\label{small-sc-vec}
		\mathrm{H}_1\big(N(f^\ast\mathcal{U}_{0})\big)\rightarrow \mathrm{H}_1\big(N(f^\ast\mathcal{U}_{1})\big)\rightarrow\cdots\rightarrow \mathrm{H}_1\big(N(f^\ast\mathcal{U}_{n})\big)
		\end{equation}
		All connecting maps in the above module are surjections.
	\end{itemize}
\end{theorem}

The above result implies that, as we proceed forward through the multiscale mapper, no
new homology classes are born. They can only die. Consequently, all bar codes in the
persistence diagram of the $H_1$-persistence module induced by it have the left endpoint at $0$.

\section{Analysis of persistent $H_1$-classes}
\label{sec:persistentH1}

Using the language of persistent homology, the results in the previous section imply
that one dimensional homology classes can die in the nerves,
but they cannot be born. In this section, we analyze further to identify
the classes that survive. The distinction among the classes is made via
a notion of `size'. Intuitively, we show that the classes with `size' much larger
than the `size' of the cover survive. The `size' is defined with the pseudometric
that the space $X$ is assumed to be equipped with. Precise statements are made 
in the subsections.

\subsection{$H_1$-classes of nerves of pseudometric spaces}
\label{H1size-sec}
Let $(X,d)$ be a pseudometric space, that is, $d$ satisfies the axioms of a metric except
that $d(x,x')=0$ may not necessarily imply $x=x'$. Assume $X$ to be compact as before. We define a `size' for a homology class that reflects how big the smallest generator in the class is in the metric $d$.

\begin{definition}
	The size $s(X')$ of a subset $X'$ of the pseudometric space $(X,d)$ is defined to be its diameter, that is, $s(X')=\sup_{x,x'\in X'\times X'} d(x,x')$. The size of a class
	$c\in H_k(X)$ is defined as $s(c)=\inf_{z\in c} s(z)$.
\end{definition}

\begin{definition}
	A set of  $k$-cycles $z_1,z_2,\ldots, z_n$ of $H_k(X)$ is called
	a generator basis if the classes $[z_1], [z_2],\ldots, [z_n]$
	together form a basis of $H_k(X)$. It is called a minimal generator basis if $\Sigma_{i=1}^n s(z_i)$ is minimal among all generator bases.
\end{definition}

%\tamal{I assumed $X$ to be compact mainly because otherwise a minimal generator basis
%	may not exist. Can you verify this? If true, we need to point out this.}

\myparagraph{Lebesgue number of a cover.} Our goal is to characterize the classes in the nerve of $\mathcal{U}$ with respect to the sizes of their preimages in $X$ via the map $\phi_\mathcal{U}$ where  $\mathcal U$ is assumed to be path connected. The Lebesgue number  of such a cover $\mathcal U$ becomes useful in this characterization. It is the largest number $\lambda(\mathcal U)$ so that any subset of $X$ with size at most $\leb({\mathcal U})$ is contained in at least one element of $\mathcal U$. Formally,
$$
\leb({\mathcal U})= \sup \{\delta\,|\, \forall X'\subseteq X \mbox{ with } s(X')\leq \delta, \exists U_\alpha \in {\mathcal U} \mbox{ where } U_\alpha\supseteq X'\}
$$

In the above definition, we can assume $X'$ to be path-connected because if it were not, then a connected superset containing all components of $X'$ is contained in $U_\alpha$ because $U_\alpha$ is path connected itself.
We observe that a homology class of size no more than $\leb({\mathcal U})$ cannot survive in the nerve. Further, the homology classes whose sizes are significantly larger than the maximum size of a cover do necessarily survive where we define the maximum size of a cover as $s_{max}(\mathcal U):= \max_{U\in {\mathcal U}} \{s(U)\}$.

Let $z_1,z_2,\ldots,z_g$ be a non-decreasing sequence of the generators with respect to their sizes in a minimal generator basis of $H_1(X)$. Consider the map
$\phi_{\mathcal U}: X \rightarrow |N({\mathcal U})|$ as introduced in Section~\ref{sec:H1}. %Assume that $\mathcal U$ is path connected. 
We have the following result.

\begin{theorem}
	Let $\mathcal U$ be a path-connected cover of $X$. 
	\begin{itemize}
		\item[i.]
		% Let $\ell\in [1,g]$ be the smallest integer so that $s(z_\ell) > \leb(\mathcal U)$. If $\ell\not= 1$, the class $\phi_{{\mathcal U}*}[z_j]=0$ for $j=1,\ldots,\ell-1$. Moreover, the classes $\{\phi_{{\mathcal U}*}[z_j]\}_{j=\ell,\ldots,g}$
		%generate $H_1(|N({\mathcal U})|)$.
		
		Let $\ell=g+1$ if $\lambda(\mathcal U) > s(z_g)$. Otherwise, let $\ell\in[1,g]$ be the smallest integer so that $s(z_\ell)> \lambda(\mathcal U)$. If
		$\ell\not=1$, the class $\bar{\phi}_{{\mathcal U}*}[z_j]=0$ for $j=1,\ldots, \ell-1$.
		Moreover, if $\ell\not=g+1$, the			
		classes $\{\bar{\phi}_{{\mathcal U}*}[z_j]\}_{j=\ell,\ldots,g}$ 
		generate $H_1(N(\mathcal U))$.
		%$H_1(N(f^*{\mathcal U}))$. 
		
		\item[ii.] The classes $\{\bar{\phi}_{{\mathcal U}*}[z_j]\}_{j=\ell',\ldots, g}$ are linearly independent where
		$s(z_{\ell'})> 4s_{max}({\mathcal U})$.
	\end{itemize}
	\label{H1prop-mapper}
\end{theorem}

The result above says that only the classes of $H_1(X)$ generated by generators of large enough size survive in the nerve.
To prove this result, we use a map $\rho$ that sends each $1$-cycle in $N({\mathcal U})$
to a $1$-cycle in $X$. We define a chain map 
$\rho: {\cal C}_1(N({\mathcal U}))\rightarrow {\cal C}_1(X)$ among one dimensional chain groups
as follows \footnote{We note that the high level framework of defining such a chain map and analyzing what it does to homologous cycles is similar to the work by Gasparovic et al. \cite{GGP16}. The technical details are different.}. It is sufficient to exhibit the map for an elementary chain 
% to characterize the 1-dimensional topological information encoded in the intrinsic \v{C}ech complex filtration of a \emph{metric graph}. The technical details are different.}. It is sufficient to show the map for an elementary chain
of an edge, say $e=\{u_{\alpha},u_{\alpha'}\} \in {\cal C}_1(N({\mathcal U}))$.
Since $e$ is an edge in $N({\mathcal U})$, the two cover elements
$U_{\alpha}$ and $U_{\alpha'}$ in $X$ have a common intersection. Let $a\in U_\alpha$ and $b\in U_{\alpha'}$ be two points that are arbitrary but fixed for $U_{\alpha}$ and
$U_{\alpha'}$ respectively. 
Pick a path $\xi(a,b)$ (viewed as a singular chain) in the union of $U_\alpha$ and $U_{\alpha'}$ which is path connected as both $U_\alpha$ and $U_{\alpha'}$ are. Then, define $\rho(e)=\xi(a,b)$. 
The following properties of $\phi_{\mathcal U}$ and $\rho$ turn out to be useful.

% 	the cycle $\gamma$  when pushed back by	$\rho$ and then pushed forward by $\phi_{\mathcal U}$ remains in the same class, that is,
\begin{proposition}
	Let $\gamma$ be any $1$-cycle in $N({\mathcal U})$. Then, $[\phi_{\mathcal U}(\rho(\gamma))]=[|\gamma|]$.
	\label{P1-mapper}
\end{proposition}
%\begin{proof}

\noindent{{\bf {Proof.~}}}
Let $e=(u_\alpha,u_{\beta})$ be an edge in $\gamma$ with $u_\alpha$ and $u_\beta$ corresponding to $U_\alpha$ and $U_\beta$ respectively. 
Let $a$ and $b$ be the corresponding fixed points for set $U_\alpha$ and $U_\beta$ respectively. Consider the path $\rho(e)=\xi(a,b)$ in $X$ as constructed above, and set $\gamma_{a,b}=\phi_{\mathcal U}(\xi(a,b))$ to be the image of $\rho(e)$ in $|N(\mathcal{U})|$. 
See Figure \ref{fig:edgedeform} for an illustration. 
Given an oriented path $\ell$ and two points $x,y\in \ell$, we use $\ell[x,y]$ to denote the subpath of $\ell$ from $x$ to $y$. 
For a point $x \in X$, for simplicity we set $\hat{x} = \phi_{\mathcal{U}}(x)$ to be its image in $|N(\mathcal{U})|$. 

\ifpaper

Now, let $w \in \rho(e)$ be a point in $U_\alpha \cap U_\beta$, and $\hat{w} = \phi_{\mathcal U} (w)$ be its image in $\gamma_{a,b}$. 
We have the following observations. 
First, any point from $\gamma_{a,b}[\hat a, \hat w]$ is contained in a simplex in $N(\mathcal{U})$ incident on $u_\alpha$. Similarly, any point from $\gamma_{a,b}[\hat w, \hat b]$ is contained in a simplex in $N(\mathcal{U})$  
\begin{wrapfigure}{r}{0.55\textwidth}
	%\vspace*{-0.2in}
	\begin{tabular}{cc}
		\includegraphics[height=2.5cm]{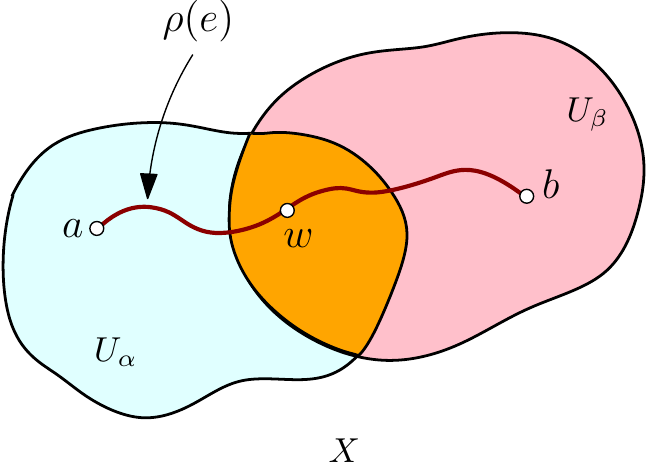} &\includegraphics[height=2.5cm]{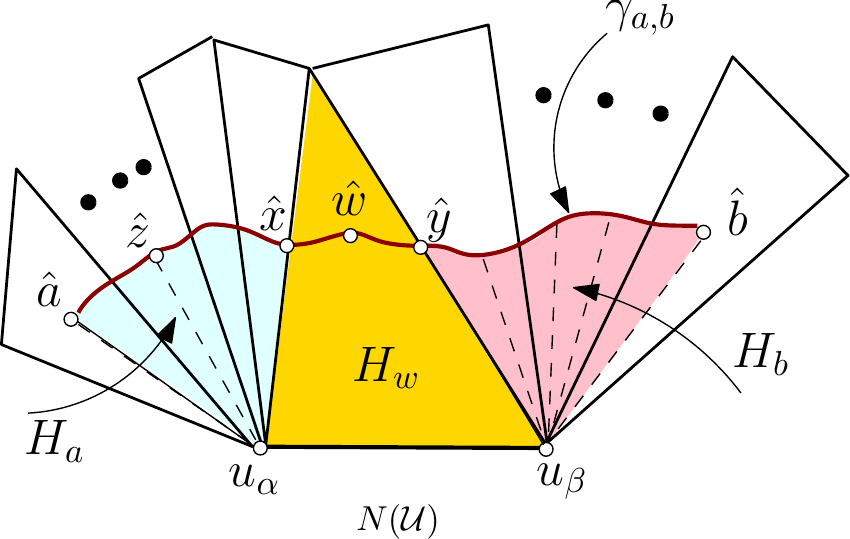}\\
		(a) & (b)
	\end{tabular}
	%\vspace*{-0.2in}
	\caption{Illustration for proof of Proposition \ref{P1-mapper}. \label{fig:edgedeform}}
	%\vspace{-0.2in}
\end{wrapfigure}
incident on $u_\beta$.
%Furthermore, let $\sigma_w \in N(\mathcal{U})$ be the highest-dimensional simplex containing $\hat{w}$. Then $e$ is necessarily a face (edge) of simplex $\sigma_w$. 
These claims simply follow from the facts that $\rho(e)[a, w] \subset U_\alpha$ and $\rho(e)[w,b]\subset U_\beta$. 
Furthermore, let $\sigma_w \in N(\mathcal{U})$ be the lowest-dimensional simplex containing $\hat w$. Depending on the partition of unity that induces the map $\phi_{\mathcal{U}}: X \to |N(\mathcal{U})|$, it is possible that $u_\alpha$ and $u_\beta$ are {\bf not} vertices of $\sigma_w$. However, as $w$ is contained in each of the cover element from $\mathcal{U}$ corresponding to the vertices of $\sigma_w$, and $w \in U_\alpha \cap U_\beta$, it must be contained in the common intersection of all these cover elements; thus there must exist simplex $\bar{\sigma}_w \in N(\mathcal{U})$ spanned by $Vert(\sigma_w) \cup \{ u_\alpha, u_\beta\}$. 
%These claims simply follow from the facts that $\rho(e)[a, w] \subset U_\alpha$, $\rho(e)[w,b]\subset U_\beta$ and $w\in U_\alpha\cap U_\beta$. 

To this end, let $\gamma_{a,b}[\hat x, \hat y]$ be the maximal subpath of $\gamma_{a,b}$ containing $\hat w$ that is contained within $|\bar{\sigma}_w|$. 
We assume that $\hat x \neq \hat y$ -- The case where $\hat x = \hat y$ can be handled by a perturbation argument which we omit here. 

Since the path $\gamma_{a,b}[\hat a, \hat x]$ is contained within a union of simplices all incident to the vertex $u_\alpha$, one can construct a homotopy $H_a$ that takes $\gamma_{a,b}[\hat a, \hat x]$ to $u_\alpha$ under which any point $\hat z \in \gamma_{a,b}[\hat a, \hat x]$ moves monotonically along the segment $\hat z u_\alpha$ within the geometric realization of the simplex containing both $\hat z$ and $u_\alpha$. 
%(Specifically, we can decompose $\gamma_{a,b}[\hat a, \hat x])$ into a set of pieces each of which is contained in a single simplex incident on $u_\alpha$. We then construct a homotopy within each of this simplex and glue all together to form $H_\alpha$. )
See Figure \ref{fig:edgedeform} (b) where we draw a simple case for illustration. 
Similarly, there is a homotopy $H_b$ that takes $\gamma_{a,b}[\hat y, \hat b]$ to $u_\beta$ under which any point $\hat z \in \gamma_{a,b}[\hat y, \hat b]$ moves monotonically along the segment $\hat z u_\beta$. 
Finally, for the middle subpath $\gamma_{a,b}[\hat x, \hat y]$, since it is within simplex $|\bar\sigma_w|$ with $e = (u_\alpha, u_\beta)$ being an edge of it, we can construct homotopy $H_w$ that takes $\gamma_{a,b}[\hat x, \hat y]$ to $u_\alpha u_\beta$ under which $\hat x$ and $\hat y$ move monotonically along the segments $\hat x u_\alpha$ and $\hat y u_\beta$ within the geometric realization of simplex $\bar\sigma_w$, respectively. 
Concatenating $H_a$, $H_w$ and $H_b$, we obtain a homotopy $H_{\alpha,\beta}$ taking $\gamma_{a,b}$ to $|e|$. 
Therefore, a concatenation of these homotopies $H_{\alpha,\beta}$ considered over all edges in $\gamma$, brings $\phi_{\mathcal U}(\rho(\gamma))$ to $|\gamma|$ with a homotopy in $|N({\mathcal U})|$. Hence, their homology classes are the same.
\else
Now, let $w \in \rho(e)$ be a point in $U_\alpha \cap U_\beta$, and $\hat{w} = \phi_{\mathcal U} (w)$ be its image in $\gamma_{a,b}$. Furthermore, let $\sigma_w \in N(\mathcal{U})$ be the lowest-dimensional simplex containing $\hat w$. While $u_\alpha$ and $u_\beta$ may not be vertices of $\sigma_w$, we can show that $Vert(\sigma_w) \cup \{ u_\alpha, u_\beta\}$ must span a simplex $\bar{\sigma}_w$, in $N(\mathcal{U})$ (see full version). 
Let $\gamma_{a,b}[\hat x, \hat y]$ be the maximal subpath of $\gamma_{a,b}$ containing $\hat w$ that is contained within $|\bar{\sigma}_w|$. 
One can construct a homotopy $H_a$ that takes $\gamma_{a,b}[\hat a, \hat x]$ to $u_\alpha$ under
\begin{wrapfigure}{r}{0.55\textwidth}
	\vspace*{-0.05in}
	\begin{tabular}{cc}
		\includegraphics[height=2.5cm]{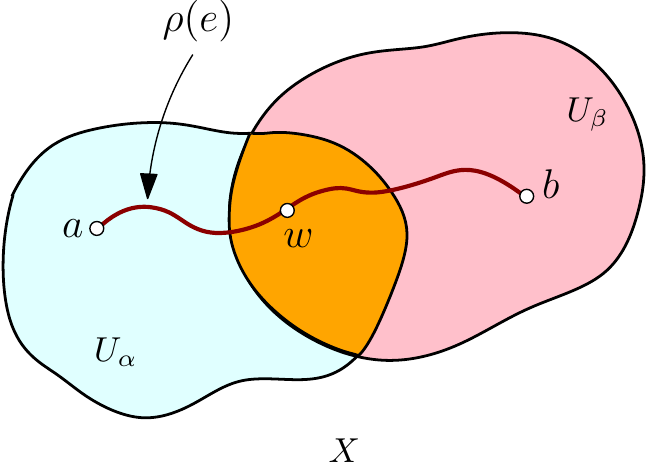} &\includegraphics[height=2.5cm]{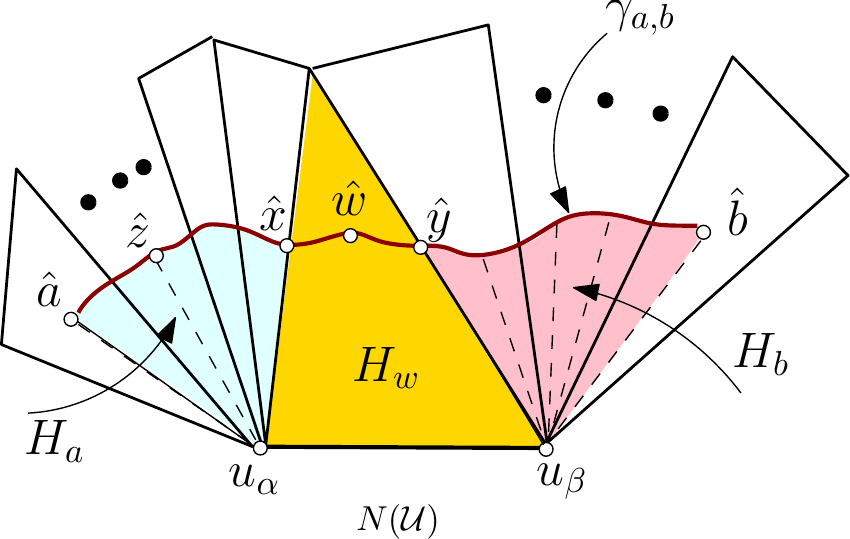}\\
		(a) & (b)
	\end{tabular}
	\vspace*{-0.1in}
	\caption{Illustration for proof of Proposition \ref{P1-mapper}. \label{fig:edgedeform}}
	\vspace*{-0.15in}
\end{wrapfigure}
which any point $\hat z \in \gamma_{a,b}[\hat a, \hat x]$ moves monotonically along the segment $\hat z u_\alpha$ within the geometric realization of the simplex containing both $\hat z$ and $u_\alpha$. See the details in the full version.
Similarly, there is a homotopy $H_b$ that takes $\gamma_{a,b}[\hat y, \hat b]$ to $u_\beta$ under which any point $\hat z \in \gamma_{a,b}[\hat y, \hat b]$ moves monotonically along the segment $\hat z u_\beta$. 
Finally, for the middle subpath $\gamma_{a,b}[\hat x, \hat y]$, since it is within simplex $\bar\sigma_w$ with $e = (u_\alpha, u_\beta)$ being an edge of it, we can construct a homotopy $H_w$ that takes $\gamma_{a,b}[\hat x, \hat y]$ to $|u_\alpha u_\beta|$ under which $\hat x$ and $\hat y$ move monotonically along the segments $\hat x u_\alpha$ and $\hat y u_\beta$ within the geometric realization of simplex $\bar{\sigma}_w$, respectively. 
Concatenating $H_a$, $H_w$ and $H_b$, we obtain a homotopy $H_{\alpha,\beta}$ taking $\gamma_{a,b}$ to $|e|$. 
A concatenation of these homotopies $H_{\alpha,\beta}$ considered over all edges in $\gamma$, brings $\phi_{\mathcal U}(\rho(\gamma))$ to $|\gamma|$ with a homotopy in $|N({\mathcal U})|$. Hence, their homology classes are the same. 
%\end{proof}
\fi
\qed
\begin{proposition}
	Let $z$ be a $1$-cycle in ${\cal C}_1(X)$.
	Then, $[\phi_{\mathcal U}(z)]=0$ if $\leb({\mathcal U})>s(z)$.
	\label{P2-mapper}
\end{proposition}
\ifpaper
\begin{proof}
	It follows from the definition of the Lebesgue number that there exists a cover element $U_\alpha\in {\mathcal U}$ so that $z\subseteq U_\alpha$ because $s(z)< \lambda(\mathcal U)$. We claim that there is a homotopy equivalence that sends $\phi_{\mathcal U}(z)$ to a vertex in $N(\mathcal U)$ and hence $[\phi_{\mathcal U}(z)]$ is trivial.
	
	Let $x$ be any point in $z$. Recall that $\phi_{\mathcal U}(x)=\Sigma_i \varphi_i(x)u_{\alpha_i}$. Since $U_\alpha$ has a common intersection with
	each $U_{\alpha_i}$ so that
	$\varphi_{\alpha_i}(x)\not=0$, we can conclude that $\phi_{\mathcal U}(x)$ is contained in a simplex with the vertex $u_\alpha$. Continuing this argument with all points of $z$, we observe that $\phi_{\mathcal U}(z)$ is contained in simplices that share the vertex $u_\alpha$. It follows that there is a homotopy that sends $\phi_{\mathcal U}(z)$ to
	$u_\alpha$, a vertex of $N(\mathcal U)$.
\end{proof}
%\else
%Proof is deferred to the full version.
\fi
\begin{proof}
	[Proof of Theorem~\ref{H1prop-mapper}]
	
	\noindent
	Proof of (i): 
	By Proposition~\ref{P2-mapper}, we have $\phi_{{\mathcal U}*}[z]=[\phi_{\mathcal U}(z)]=0$ if $\lambda(\mathcal U)> s(z)$. This establishes the first part of the assertion because
	$\bar{\phi}_{{\mathcal U}*}=\iota \circ \phi_{{\mathcal U}*}$ where $\iota$ is
	an isomorphism between the singular homology of $|N(\mathcal U)|$ and the simplicial homology of $N(U)$. To see the second part, notice that
	$\bar{\phi}_{{\mathcal U}*}$ is a surjection by Theorem~\ref{surj-map-thm}. Therefore, the classes
	$\bar{\phi}_{{\mathcal U}*}(z)$ where $\lambda(\mathcal U)\not> s(z)$ contain a basis
	for $H_1(N(\mathcal U))$. Hence they generate it.\\

	%Suppose on the contrary there is a $1$-cycle 
	%in $|N(X,{\mathcal U})|$ whose class cannot be written as the linear combination
	%of $\{{\phi_{\mathcal U}[z_j]}_{j=\ell,\dots,g}\}$ as claimed. Let $\gamma$ be the $1$-cycle in $N(X,{\mathcal U})$ where $|\gamma|$ belongs to the same class. Such a $1$-cycle exists by Proposition~\ref{embed-prop}. Consider the
	%cycle $\rho(\gamma)$ in $X$. The class of this
	%cycle must be a
	%linear combination of the basis elements, that is, for $\alpha_i\in\{0,1\}$,
	%$[\rho(\gamma)]=\Sigma_{i=1}^g \alpha_i[z_i]$.
	%It follows that
	%$[\phi_{\mathcal U}\rho(\gamma)]=\Sigma_{i=1}^g \alpha_i [\phi(z_i)]$.
	%Because of Proposition~\ref{P2-mapper},
	%we have $[\phi_{\mathcal U}(z_i)]=0$ for $i< \ell$ where $s(z_\ell)>\delta({\mathcal U})> s(z_{\ell-1})$.
	%Therefore, 
	%$[\phi_{\mathcal U}\rho(\gamma)]=\Sigma_{i=\ell}^g \alpha_i [\phi_{\mathcal U}(z_i)]$,
	%where $s(z_\ell)>\delta({\mathcal U})> s(z_{\ell-1})$.
	%By Proposition~\ref{P1-mapper},
	%$[|\gamma|]=[\phi_{\mathcal U}\rho(\gamma)]=\Sigma_{i=\ell}^g \alpha_i[\phi_{\mathcal U}(z_i)]$
	%contradicting our assumption.\\
	
	\noindent
	Proof of (ii): Suppose on the contrary, there is a subsequence
	$\{\ell_1,\ldots,\ell_t\}\subset \{\ell',\ldots,g\}$ such that
	$\Sigma_{j=1}^t [\phi_{\mathcal U}(z_{\ell_j})] =0$. Let $z=\Sigma_{j=1}^t \phi_{\mathcal U}(z_{\ell_j})$. Let $\gamma$ be a $1$-cycle in $N({\mathcal U})$ so that
	$[z]=[|\gamma|]$ whose existence is guaranteed by Proposition~\ref{embed-prop}.
	It must be the case that there is a $2$-chain
	$D$ in $N({\mathcal U})$ so that $\partial D=\gamma$.
	% \yusu{For this step: Even if for the simplicial gamma $\gamma$, its underlying curve has trivial homology, why would $\gamma$ as a simplicial cycle necessarily has trivial homology in the simplicial homology? Don't we need then that $\gamma$ and $\iota(|\gamma|)$ have to be homologous to be able to say this? } 
	Consider a triangle $t=\{u_{\alpha_1},u_{\alpha_2},u_{\alpha_3}\}$ contributing to $D$. Let $a_i'=\phi_{\mathcal U}^{-1}(u_{\alpha_i})$. Since $t$ appears 
	in $N({\mathcal U})$, the covers $U_{\alpha_1}, U_{\alpha_2}, U_{\alpha_3}$ containing
	$a_1'$, $a_2'$, and $a_3'$ respectively have a common intersection in $X$. This also
	means that each of the paths $a_1'\leadsto a_2'$, $a_2'\leadsto a_3'$, $a_3'\leadsto a_1'$
	has size at most $2 s_{max}(\mathcal U)$. 
	Then, $\rho(\partial t)$ is mapped to a $1$-cycle in $X$
	of size at most $4s_{max}(\mathcal U)$. 
	It follows that $\rho(\partial D)$ can be written as a linear
	combination of cycles of size at most $4s_{max}(\mathcal U)$. Each of the $1$-cycles
	of size at most $4s_{max}(\mathcal U)$ is generated by basis elements $z_1,\ldots,z_k$
	where $s(z_k)\leq 4s_{max}(\mathcal U)$. Therefore, the class of $z'=\phi_{\mathcal U}(\rho(\gamma))$ is generated by a linear combination of
	the basis elements whose preimages have size at most $4s_{max}(\mathcal U)$. The class
	$[z']$ is same as the class $[|\gamma|]$ by Proposition~\ref{P1-mapper}. But, by assumption $[|\gamma|]=[z]$ is
	generated by a linear combination of the basis elements whose sizes
	are larger than $4s_{max}(\mathcal U)$ reaching a contradiction. 
\end{proof}

\subsection{$H_1$-classes in Reeb space}
In this section we prove an analogue of Theorem~\ref{H1prop-mapper} for Reeb spaces, which to our knowledge is new.
The Reeb space of a function $f:X\rightarrow Z$, denoted $R_f$, is the quotient of $X$ under the equivalence relation $x\sim_f x'$ if and only if $f(x)=f(x')$ and there exists a continuous path $\gamma\in\Gamma_X(x,x')$ such that $f\circ\gamma$ is constant. The induced quotient map  is denoted $q:X\rightarrow R_f$ which is of course surjective. We show that $q_*$ at the homology level is also surjective for $H_1$ when the codomain $Z$ of $f$ is a metric space. In fact, we prove a stronger statement: only `vertical' homology classes (classes with strictly positive size) survive in a Reeb space which extends the result of Dey and Wang~\cite{DW13} for Reeb graphs.

Let $\mathcal V$ be a path-connected cover of $R_f$. This induces a pullback cover denoted $\mathcal U = \{U_\alpha\}_{\alpha\in A}=\{q^{-1}(V_\alpha)\}_{\alpha\in A}$ on $X$. Let $N(\mathcal U)$ and $N(\mathcal V)$ denote the corresponding 
nerve complexes of $\mathcal{U}$ and $\mathcal{V}$ respectively. It is easy to see that $N(\mathcal U)=N(\mathcal V)$ because $U_\alpha\cap U_{\alpha'}\not=\emptyset$ if and only if $V_\alpha\cap V_{\alpha'}\not=\emptyset$. There are nerve maps $\phi_{\mathcal V}: R_f\rightarrow |N(\mathcal V)|$ and $\phi_{\mathcal U}: X\rightarrow |N(\mathcal U)|$ so that the following holds:
\begin{proposition}
	Consider the sequence $X\stackrel{q}{\rightarrow} R_f(X)\stackrel{\phi_{\mathcal V}}{\rightarrow} |N(\mathcal V)|=|N(\mathcal U)|$. Then, $\phi_{\mathcal U}=\phi_{\mathcal V}\circ q$.
	\label{Reeb-commut-prop}
\end{proposition}
\ifpaper
\begin{proof}
	Consider a partition of unity $\{\varphi_\alpha\}_{\alpha\in A}$ subordinate to $\mathcal V= \{V_\alpha\}_{\alpha\in A}$. Without loss of generality,
	one can assume $\mathcal V$ to be locally finite because $X$ is paracompact. Then, consider the partition of unity subordinate to $\mathcal U=\{U_{\alpha}\}_{\alpha\in A}$ given by
	$\varphi'_{\alpha}(q^{-1}(x))=\varphi_\alpha(x)$. Let $\phi_{\mathcal V}$ and $\phi_{\mathcal U}$ be the nerve maps corresponding to the partition of unity of $\varphi_\alpha$ and $\varphi'_{\alpha}$ respectively. Then, $\phi_{\mathcal U}(x)=\phi_{\mathcal V}(q(x))$ proving the claim.
\end{proof}
%\else Proof is deferred to the full version.
\fi

Let the codomain of the function $f: X\rightarrow Z$ be a {\em metric space} $(Z,d_Z)$. We first impose a pseudometric on $X$ induced by $f$; the one-dimensional version of this pseudometric is similar to the one used in \cite{BGW14} for Reeb graphs. 
Recall that given two points $x,x'\in X$ we denote by $\Gamma_X(x,x')$ the set of all continuous paths $\gamma:[0,1]\rightarrow X$ such that $\gamma(0)=x$ and $\gamma(1)=x'.$

\begin{definition}
	We define a pseudometric $d_f$ on $X$ as follows: for $x,x'\in X$, 
	$$d_f(x,x'):=\inf_{\gamma\in \Gamma_X(x,x')}\diam_Z(f\circ\gamma).$$
\end{definition}
\begin{proposition}
	$d_f:X\times X\rightarrow \mathbb{R}_+$ is a pseudometric.
\end{proposition}
\ifpaper
\begin{proof}
	Symmetry, non-negativity,  and the fact that $d_f(x,x)=0$ for all $x\in X$ are evident. We prove the triangle inequality. We will use the following claim whose proof we omit.
	\begin{claim}\label{claim:diam}
		For all $A,B \subseteq Z$ with $A\cap B\neq \emptyset$ we have $\diam_Z(A\cup B)\leq \diam_Z(A) + \diam_Z(B).$
	\end{claim}
	
	Assume $x,x',x''\in X$ are such that $a = d_f(x,x')$ and $a'=d_f(x',x'')$. Fix any $\varepsilon>0$. Choose $\gamma \in \Gamma_X(x,x')$ and $\gamma'\in \Gamma_X(x',x'')$ such that $\diam_Z(f\circ\gamma)<a+\frac{\varepsilon}{2}$ and  $\diam_Z(f\circ\gamma')<a'+\frac{\varepsilon}{2}$. Now consider the curve $\gamma'':[0,1]\rightarrow X$ defined by concatenating $\gamma$ and $\gamma'$ so that $\gamma''\in \Gamma_X(x,x'').$ Then, by the above claim, we have
	$$d_f(x,x'')\leq \diam_Z(f\circ\gamma'') = \diam_Z(\{f\circ\gamma\}\cup\{f\circ\gamma'\})\leq a+a' +\varepsilon.$$
	The claim is obtained by letting $\varepsilon\rightarrow 0.$
	
	% \begin{proof}[Proof of Claim \ref{claim-diam}]
	% Pick  any $x,x'\in A\cup B$. Then, either
\end{proof}
%\else
%Proof is deferred to the full version.
\fi
Similar to $X$, 
we endow $R_f$ with a distance $\tilde{d}_f$ that descends via the map $q$: for any equivalence classes $r,r'\in R_f$, pick $x,x'\in X$ with $r =q(x)$ and $r'=q(x')$, then define
$$\tilde{d}_f(r,r'):=d_f(x,x').$$
The definition does not depend on the representatives $x$ and $x'$ chosen. In this manner we obtain the pseudometric space $(R_f,\tilde{d}_f).$
Let $z_1,\ldots,z_g$ be a minimal generator basis of $H_1(X)$ defined with respect to the pseudometric $d_f$ and $q: X\rightarrow R_f$ be the quotient map.

\begin{theorem}
	Let $\ell\in[1,g]$ be the smallest integer so that $s(z_\ell)\not=0$. If no such $\ell$ exists,
	$H_1(R_f)$ is trivial, otherwise,  $\{[q(z_i)]\}_{i=\ell,\ldots g}$ is a basis for $H_1(R_f)$.
	\label{RS-thm}
\end{theorem}
\begin{proof}

	Consider the sequence $X\stackrel{q}{\rightarrow}R_f\stackrel{\phi_{\mathcal V}}{\rightarrow} |N(\mathcal V)|$ where $\mathcal V$ is a cover of $R_f$.
	\ifpaper
	\begin{claim} 
		$q_*$ is a surjection.
	\end{claim}
	
	\begin{proof}
		Let $z_1',z_2',\ldots,z_{g'}'$ be a minimal generator basis of $H_1(R_f)$ of the metric space $(R_f,\tilde{d}_f)$. Observe that
		$s(z_i')\not=0$ for any $i\in[1,g']$ because otherwise we have a $z_j'$ for some $j\in [1,g']$ whose any two distinct points $x,x'\in z_j'$ satisfy 
		$f(q^{-1}(x))=f(q^{-1}(x'))$ and $q^{-1}(x)$ and $q^{-1}(x')$ are path connected in $X$. This is impossible by the definition of $R_f$.	
		
		Without loss of generality,
		assume that	$\mathcal V$ is fine enough so that it
		satisfies $0<s_{\mathrm{max}}(\mathcal V)\leq \delta$ where $\delta=\frac{1}{4}\min\{s(z_i')\}$. Since $\delta>0$ due to the observation in the previous paragraph, such a cover exists.
		Then, by applying Theorem~\ref{H1prop-mapper}(ii), we obtain that 
		$[\phi_{\mathcal V}(z_i')]_{i=1,\ldots,g'}$ are linearly independent in
		$H_1(|N(\mathcal V)|$. It follows that $\phi_{{\mathcal V}*}$ is injective. It is surjective too by Proposition~\ref{phi-surject-prop}. Therefore, $\phi_{{\mathcal V}*}$ is an isomorphism.
		
		Let $\mathcal U$ be the pullback cover of $\mathcal V$. Then, we have $\phi_{{\mathcal U}*}=\phi_{{\mathcal V}*}\circ q_*$ (Proposition~\ref{Reeb-commut-prop}) where $\phi_{{\mathcal U}*}$ is a surjection and $\phi_{{\mathcal V}*}$ is an ismorphism. It follows that $q_*$ is a surjection.
	\end{proof}
	\else
	It is shown in the full version that $q_*$ is a surjection for $H_1$-homology.
	\fi
	
	By the previous claim, $\{[q(z_i)]\}_{i=1,\ldots,g}$ generate $H_1(R_f)$. First, assume that $\ell$ as stated in the theorem exists.
	Let the cover 
	$\mathcal V$ be fine enough so that 
	$0<s_{\mathrm{max}}(\mathcal U)\leq \delta$ where $\delta=\frac{1}{4}\min\{s(z_i)\,|\, s(z_i)\not=0\}$. 
	Then, by applying Theorem~\ref{H1prop-mapper}(ii), we obtain that 
	$[\phi_{\mathcal U}(z_i)]_{i=\ell,\ldots,g}$ are linearly independent in
	$H_1(|N(\mathcal U)|)=H_1(|N(\mathcal V)|$. Since $[\phi_{\mathcal U}(z_i)]=[\phi_{\mathcal V}\circ q(z_i)]$ by Proposition~\ref{Reeb-commut-prop}, $\{[q(z_i)]\}_{i=\ell,\ldots,g}$ are linearly independent in $H_1(R_f)$. But, $[q(z_i)]=0$ for $s(z_i)=0$ and $\{[q(z_i)]\}_{i=1,\ldots,g}$ generate $H_1(R_f)$.
	Therefore, $\{[q(z_i)]\}_{i=\ell,\ldots,g}$ is a basis.
	In the case when $\ell$ does not exist, we have $s(z_i)=0$ for every $i\in[1,g]$. Then, $[q(z_i)]=0$ for every $i$ rendering $H_1(R_f)$ trivial.
\end{proof}   

\subsection{Persistence of $H_1$-classes in mapper and multiscale mapper}

To apply the results for nerves in section~\ref{H1size-sec} to mappers and multiscale mappers, the Lebesgue number of the pullback covers of $X$ becomes important. The following observation in this respect is useful. Remember that the size of a subset in $X$ and hence the cover elements are measured with respect to the pseudometric $d_f$.

\begin{proposition}
	Let $\mathcal U$ be a cover for the codomain $Z$. Then, the pullback cover $f^*{\mathcal U}$ has Lebesgue number $\leb({\mathcal U})$.
\end{proposition} 
\ifpaper
\begin{proof}
	Let $X'\subseteq X$ be any path-connected subset where $s(X')\leq \lambda(\mathcal U)$. Then, $f(X')\subseteq Z$ has a diameter at most $\lambda(\mathcal U)$ by the definition of size.
	Therefore, by the definition of Lebesgue number, $f(X')$ is contained in a cover element $U\in{\mathcal U}$. Clearly, a path connected component of $f^{-1}(U)$ contains $X'$ since $f$ is assumed to be continuous. It follows that there is a cover element in $f^\ast{\mathcal U}$ that contains $X'$. Since $X'$ was chosen as an arbitrary subset of size at most 
	$\lambda(\mathcal U)$, we have $\lambda(f^*{\mathcal U})\geq \lambda(\mathcal U)$.
	At the same time, it is straightforward from the definition of size that each cover element in $f^{-1}(U)$ has at most the size of $U$ for any $U\in{\mathcal U}$. Therefore, $\lambda(f^*{\mathcal U})\leq \lambda(\mathcal U)$ establishing the equality as claimed.
\end{proof} 
%\else
%Proof is deferred to the full version.
\fi

Notice that the smallest size $s_{min}(f^\ast{\mathcal U})$ of an element of the pullback cover can be arbitrarily small even if $s_{min}({\mathcal U})$ is not. However, the Lebesgue number of $\mathcal U$ can be leveraged for the mapper due to the above Proposition.

Given a cover $\mathcal U$ of $Z$, consider the mapper $N(f^*{\mathcal U})$.
Let $z_1,\ldots,z_g$ be a set of minimal generator basis for $H_1(X)$ where the metric
in question is $d_f$. Then, as a consequence of Theorem~\ref{H1prop-mapper} we have:

\begin{theorem}
	~
	\begin{itemize}
		\item[i] Let $\ell=g+1$ if $\lambda(\mathcal U) > s(z_g)$. Otherwise, let $\ell\in[1,g]$ be the smallest integer so that $s(z_\ell)> \lambda(\mathcal U)$. If
		$\ell\not=1$, the class $\phi_{{\mathcal U}*}[z_j]=0$ for $j=1,\ldots, \ell-1$.
		Moreover, if $\ell\not=g+1$, the			
		classes $\{\phi_{{\mathcal U}*}[z_j]\}_{j=\ell,\ldots,g}$ 
		generate $H_1(N(f^*{\mathcal U}))$.
		
		\item[ii] The classes $\{\phi_{{\mathcal U}*}[z_j]\}_{j=\ell',\ldots, g}$ are linearly independent where $s(z_{\ell'})> 4 s_{max}({\mathcal U})$.
		\item[iii] Consider a $H_1$-persistence module of a multiscale mapper induced by a tower of path connected covers:
		\begin{equation}
		\mathrm{H}_1\big(N(f^\ast\mathcal{U}_{\eps_0})\big)\stackrel{s_{1*}}{\rightarrow} \mathrm{H}_1\big(N(f^\ast\mathcal{U}_{\eps_1})\big)\stackrel{s_{2*}}{\rightarrow}\cdots\stackrel{s_{n*}}{\rightarrow} \mathrm{H}_1\big(N(f^\ast\mathcal{U}_{\eps_n})\big)
		\end{equation}
		Let $\hat{s}_{i*}=s_{i*}\circ s_{(i-1)*}\circ \cdots \circ \bar{\phi}_{{\mathcal U}_{\eps_0}*}$. Then, the assertions in (i) and (ii) hold for $H_1(N(f^*{\mathcal U}_{\eps_i}))$ with the map $\hat{s}_{i*}: X\rightarrow N(f^*{\mathcal U}_{\eps_i})$.
	\end{itemize}
	\label{H1pers-thm}
\end{theorem}

\begin{remark}[Persistence diagram approximation.]
	The persistence diagram of the $H_1$-persistence module considered in Theorem~\ref{H1pers-thm}(iii) contains points whose birth coordinates are exactly zero. This is because all connecting maps are surjective by (i) and thus every class is born only at the beginning. The death coordinate of a point that corresponds to a minimal basis generator of size $s$ is in between the index $\eps_i$ and $\eps_j$ where
	$s\geq 4 s_{max} ({\mathcal U}_{\eps_i})$ and $s\leq \lambda({\mathcal U}_{\eps_j})$
	because of the assertions (i) and (ii) in Theorem~\ref{H1pers-thm}.
	Assuming covers whose $\lambda$ and $s_{max}$ values are within a constant factor of each other (such as the ones described in next subsection), we can conclude that a generator of size $s$ dies at some point $cs$ for some constant $c$.
	Therefore, by computing a minimal generator basis of $N({\mathcal U}_{\eps_0})$ and computing their sizes provide a $4$-approximation to the persistence diagram of the multiscale mapper in the log scale. 
\end{remark}

\subsection{Two special covers and intrinsic \v{C}ech complex} We discuss two special covers, one can be effectively computed
and the other one is relevant in the context of the intrinsic \v{C}ech complex of a metric space. We say a cover $\mathcal U$ of a metric space $(Y,d)$ is $(\alpha,\beta)$-cover if $\alpha\leq\lambda(\mathcal U)$ and $\beta\geq s_{max}(\mathcal U)$.\\

\noindent
{\bf A $(\delta,4\delta)$-cover}: Consider a $\delta$-sample $P$ of $Y$, that is, every metric ball $B(y,\delta)$, $y\in Y$, contains a point in $P$. Observe that the cover ${\mathcal U}=\{B(p, 2\delta)\}_{p\in P}$ is a $(\delta,4\delta)$-cover for $Z$. Clearly, $s_{max}(\mathcal U)\leq 4\delta$. To determine $\lambda(\mathcal U)$, consider any subset $Y'\subseteq Y$ with $s(Y')\leq \delta$. There is a $p\in P$ so that $d_Y(p,Y')\leq \delta$. Let $y'$ be the furthest point in $Y'$ from $p$. Then, $d_Y(p,y')\leq d_Y(p,Y)+\diam(Y')\leq 2\delta$ establishing that $\lambda(\mathcal U)\geq\delta$.
%\tamal{Should we add something about computation?}

\smallskip
\noindent
{\bf A $(\delta,2\delta)$-cover}: Consider the infinite cover $\mathcal U$ of $Y$ where
${\mathcal U}=\{B(y,\delta)\}_{y\in Y}$. These are the set of all metric balls of radius $\delta$. Clearly, $s_{max}(\mathcal U)\leq 2\delta$. Any subset $Y'\subseteq Y$ with $s(Y')\leq \delta$ is contained in a ball $B(y,\delta)$ where $y$ is any point in $Y'$. This shows that 
$\lambda(\mathcal U)\geq \delta$. A consequence of this observation and Theorem~\ref{H1prop-mapper} is that the intrinsic \v{C}ech complexes satisfy some interesting property.

\begin{definition}\label{def:IC}
	Given a metric space $(Y,d_Y)$, its intrinsic \v{C}ech complex $C^{\delta}(Y)$ at scale $\delta$ is defined to be the nerve complex of the set of intrinsic $\delta$-balls $\{ B(y,\delta)\}_{y\in Y}$.
\end{definition}
\begin{observation}
	Let $C^\delta(Y)$ denote the intrinsic \v{C}ech complex of a metric space $Y$ at scale $\delta$. Let $\mathcal U$ denote the corresponding possibly infinite cover of $Y$. Let $z_1,\ldots,z_g$ be a minimal generator basis for $H_1(Y)$. Then, $\{\bar{\phi}_{{\mathcal U}*}(z_i)\}_{i=\ell,\ldots, g}$ generate $H_1(C^\delta(Y))$ if $\ell$ is the smallest integer with $s(z_\ell)> \delta$. Furthermore, $\{\bar{\phi}_{{\mathcal U}*}(z_i)\}_{i=\ell',\ldots, g}$ are linearly independent if $s(z_\ell')> 8\delta$.  
\end{observation}

\section{Higher dimensional homology groups}
\label{sec:highD}

We have already observed that the surjectivity of the map $\phi_{{\mathcal U}*}: H_1(X)\rightarrow H_1(|N({\mathcal U})|)$ in one dimensional homology does not extend to higher dimensional homology groups. This means that we cannot hope for analogues to Theorem~\ref{H1prop-mapper}(i) and Theorem~\ref{H1pers-thm} to hold for higher dimensional homology groups. However, under the assumption that $f:X\rightarrow Z$ is a continuous map from a compact space to a metric space, we can provide some characterization of the persistent diagrams of the mapper and the multiscale mapper as follows:
\begin{itemize}\denselist 
	\item We define a metric $d_\delta$ on the vertex set $P_\delta$ of $N({\mathcal U})$ where $s_{\mathrm{max}}({\mathcal U})\leq\delta$ and then show that the Gromov-Hausdorff distance between the metric spaces $(P_\delta,d_\delta)$ and $(R_f,\tilde{d}_f)$ is at most $5\delta$. The same proof also applies if we replace $(R_f,\tilde{d}_f)$ with the pseudometric space $(X,d_f)$.
	\ifpaper {} \else See the full version.\fi 
	\item Previous result implies that the persistence diagrams of
	the intrinsic \v{C}ech complex of the metric space $(X,d_f)$ and that of the metric space $(P_\delta,d_\delta)$ have a bottleneck distance of $O(\delta)$.
	This further implies that the persistence diagram of the mapper structure $N(\mathcal{U})$ (approximated with the metric space $(P_\delta,d_\delta)$ ) is close to that of the intrinsic \v{C}ech complex of the pseudometric space $(X, d_f)$; see Section \ref{subsec:mapper}. 
	\item We show that the intrinsic 
	\v{C}ech complexes of $(X,d_f)$ interleave with $\MM(\mathfrak U , f)$ thus connecting their persistence diagrams. See Section \ref{subsec:multimapper}. 
	\item It follows that the persistence diagrams of the multiscale mapper $\MM(\mathfrak U , f)$ and $(P_\delta,d_\delta)$ are close, both being close to that of $(X,d_f)$. This shows that the multiscale mapper encodes similar information as the mapper under an appropriate 
	map-induced metric. 
\end{itemize}
%\tamal{The above paragraph is to be re-written correctly and nicely.}

\ifpaper
\subsection{Gromov-Hausdorff distance between Mapper and the Reeb space}

\subsubsection{Mapper as a finite metric space}\label{GH-append}
We have already shown how to equip the Reeb space $R_f$ with a distance $\tilde{d}_f$.
%Consider a finite $\delta$-net $\{z_\alpha,\,\alpha\in A\}$ of $Z$ and let $\mathcal{U}_\delta =\{B_\delta(z_\alpha),\,\alpha\in A\}$ denote the resulting open cover of $Z$. We assume the cover is efficient in the sense that no ball $B_\delta(z_\alpha)$ is contained some other ball $B_\delta(z_{\alpha'})$ unless $\alpha=\alpha'$.  For tame $f:X\rightarrow Z$ consider now the pullback cover $\mathcal{V}_\delta=f^\ast(\mathcal{U}_\delta)$ of $X$ consisting of elements $\{V_{\alpha,i},\,i\in I_\alpha\,\mbox{and}\,\alpha\in A\}.$  Consider now the nerve $M_\delta=N(\mathcal{V}_\delta)$, and let $P_\delta$ denote the vertex set of $M_\delta$; we will denote its points by $v=(\alpha,i)$ with $\alpha\in A$ and $i\in I_\alpha.$ Denote by $E_\delta$ the edge set of $M_\delta$.

Consider a cover ${\mathcal U}_\delta$ of $Z$ whose all cover elements have size at most $\delta$, that is, ${\mathcal U}_\delta=\{U_\alpha, \alpha\in A,s(U_\alpha)\leq \delta\}$. For a continuous map $f:X\rightarrow Z$ consider now the pullback cover $\mathcal{V}_\delta=f^\ast\mathcal{U}_\delta$ of $X$ consisting of elements $\{V_{\alpha,i},\,i\in I_\alpha\,\mbox{and}\,\alpha\in A\}.$  We choose an arbitrary but distinct point $z_\alpha\in U_\alpha$ for every element $U_\alpha\in \mathcal{U}_\delta$.

Consider now the nerve $M_\delta=N(\mathcal{V}_\delta)$, and let $P_\delta$ denote the vertex set of $M_\delta$; we will denote its points by $v_{\alpha,i}$ which corresponds to the element $V_{\alpha,i}$. Denote by $E_\delta$ the edge set of $M_\delta$.

Define the vertex function $f_\delta:P_\delta\rightarrow Z$ as follows: $f_\delta(v_{\alpha,i}) := z_{\alpha}$ for each $v_{\alpha,i}\in P_\delta$. Consider the metric $d_\delta:P_\delta\times P_\delta\rightarrow \mathbb{R}_+$ given by 

$$d_\delta\big(v,v'\big):= \min\big\{\diam_Z(\{f_\delta(v_\ell)\}_{\ell=0}^n),\,\mbox{where $v_0=v$, $v_n=v'$, $(v_k,v_{k+1})\in E_\delta$ for all $k$}\big\}$$
for any $v,v'\in P_\delta$. We thus form the finite metric space $(P_\delta,d_\delta).$

% Notice that because of our definition of $f_\delta$ it follows that 
% $$d_\delta\big((\alpha,i),(\alpha',i')\big) = \min\big\{\diam_Z(\{f_\delta(v_\ell)\}_{\ell=0}^n),\,\mbox{where $v_0=v$, $v_n=v'$, $(v_k,v_{k+1})\in E_\delta$ for all $\ell$}\big\}$$ 

\begin{remark}
	Verifying that $d_\delta$  is indeed a metric requires checking that $d_\delta(v,v')=0$ implies that $v=v'$.\footnote{The triangle inequality is clear.}  If $d_\delta(v,v') = 0$ then there exist $v=v_0,\ldots,v_n=v'$ in $P_\delta$ and $z_\ast\in Z$ such that $(v_k,v_{k+1})\in E_\delta$ for all $k$ such that $f_\delta(v_k)= z_\ast$ for all $k$. If we write $v_k = v_{\alpha_k,i_k}$ for $i_k\in I_{\alpha_k}$ then this means that $z_\ast = f_\delta(v_{\alpha_k,i_k}) = z_{\alpha_k}$ for all $k.$  This means that the elements $\{V_{\alpha_k,i_k},\,k=0,\ldots,n\}$ of the pullback cover are all different path connected components of the set $f^{-1}(U_{\alpha_k}).$ This means that one cannot have $(v_k,v_{k+1})\in E_\delta$ unless $v_0=v_1=\ldots,v_n$ implying that $v=v'$.
\end{remark}

We now construct a map $p_\delta:X\rightarrow P_\delta$. In order to do this consider the set of indices $B = \{(\alpha,i),\,\alpha\in A,i\in I_\alpha\}$ into elements of the cover $\mathcal{V}_\delta = f^\ast\mathcal{U}_\delta$.  Choose any total order $>_A$ on $A$, and then declare that $(\alpha,i)>(\alpha',i')$ whenever it holds  (1) $\alpha>_A\alpha'$, or (2) in case $\alpha=\alpha'$, $i>i'$.
For any $x\in X$ let 
$p_\delta(x):=v_{\alpha,i} \mbox{ where } (\alpha,i)=\min\{\beta\in B|\,x\in V_\beta\}.$ 
%That this map is surjective follows from the assumption that the cover $\mathcal{U}_\delta$ is effective. (\facundo{details?})

Notice that $p_\delta$ is not necessarily a surjection. Since our goal is to define a correspondence between $X$ and $P_\delta$, for every $V_{\alpha,i}\in {\mathcal V}_\delta$ we choose an arbitrary point $x_{\alpha,i}\in V_{\alpha,i}$ and associate it with the vertex
$v_{\alpha,i}$.

\subsubsection{A bound on the Gromov-Hausdorff distance}
The proof of the following theorem extends to $(X,d_f)$ almost verbatim.
\begin{theorem}\label{thm:gh}
	Under the conditions above, $$d_{GH}\big((R_f,\tilde{d}_f),(P_\delta,d_\delta)\big)\leq 5\delta.$$
\end{theorem}
\begin{proof}
	Consider the correspondence $S$ between $R_f$ and $P_\delta$ defined by 
	$S:=\{(q(x),p_\delta(x)),\,x\in X\}\cup\{q(x_{\alpha,i}),v_{\alpha,i}\}$. That $S$ is indeed a correspondence follows from the fact that $q:X\rightarrow R_f$ is a surjection and the second factor in $S$ covers all vertices in $P_\delta$.
	
	\begin{claim}\label{claim:1} For all $x,x'\in X$ one has 
		$$\tilde d_f(q(x),q(x'))-\delta \leq d_\delta(p_\delta(x),p_\delta(x'))\leq \tilde d_f(q(x),q(x'))+\delta.$$
	\end{claim}
	
	%\begin{claim}\label{claim:2} For all $x,x'\in X$ one has $\tilde d_f(\pi(x),\pi(x'))\leq %d_\delta(p_\delta(x),p_\delta(x')) +2\delta.$
	%\end{claim}
	
	\begin{claim}\label{claim:2}	
		For all $x, x_{\alpha,i}\in X$ one has 
		$$\tilde d_f(q(x),q(x_{\alpha,i}))-3\delta \leq d_\delta(p_\delta(x),v_{\alpha,i})\leq \tilde d_f(q(x),q(x_{\alpha,i}))+3\delta$$
	\end{claim}
	
	\begin{claim}\label{claim:3}	
		For all $x_{\alpha,i}, x_{\alpha',i'}\in X$ one has 
		$$\tilde d_f(q(x_{\alpha,i}),q(x_{\alpha',i'}))-5\delta \leq d_\delta(v_{\alpha,i},v_{\alpha',i'})\leq \tilde d_f(q(x_{\alpha,i}),q(x_{\alpha',i'}))+5\delta$$
	\end{claim}

	Combining the three claims above we obtain that 
	$$\mathrm{dis}(S)=\sup_{x,x'\in X, y\in S(x), y'\in S(x')}\big|\tilde d_f(q(x),q(x'))-d_\delta(y,y')\big|\leq 5\delta$$ thus finishing the proof.
	\end{proof}
	
	\begin{proof}[Proof of Claim \ref{claim:1}]
		We prove the upper bound. The proof for the lower bound is similar.
		Assume that $\tilde d_f(q(x),q(x'))<\eta$ for some $\eta>0$ and let $\gamma\in\Gamma_X(x,x')$ be s.t. $\diam_Z(f\circ\gamma)\leq \eta$. Consider the set of vertices $Q:=\{p_\delta(\gamma(t)),\,t\in[0,1]\}\subset P_\delta$. This set consists of a finite sequence of vertices $v_{\alpha_\ell,i_\ell}$ for $\ell=0,1,\ldots, N$, for some positive integer $N$. Notice that $f_\delta(Q) = \{z_{\alpha_\ell},\,\ell=0,1,\ldots,N\}$ and by construction we can assume that $(v_{\alpha_\ell, i_\ell},v_{\alpha_{\ell+1},i_{\ell+1}})\in E_\delta$ for each $\ell.$
		
		Now, for each $\ell\in\{0,\ldots,N\}$ there exists $t_\ell\in[0,1]$ such that $\gamma(t_\ell)\in V_{\alpha_\ell,i_\ell}$, which means that $f(\gamma(t_\ell))\in U_{\alpha_\ell}.$ But $z_{\alpha_\ell}\in U_{\alpha_\ell}$ so that then 
		$f(Q) \subseteq \bigcup_{\ell=0}^N U_{\alpha_\ell} .$ At the same time,
		$\bigcup_{t\in[0,1]}(f(\gamma(t)))\subseteq \bigcup_{\ell=0}^N U_{\alpha_\ell}$.
		Hence, $$\delta+\eta\geq \delta + \diam_Z(f\circ \gamma)\geq \diam_Z(f(Q))\geq d_\delta(p_\delta(x),p_\delta(x')).$$
		The proof of the upper bound follows by letting $\eta\rightarrow \tilde d_f(q(x),q(x')).$
	\end{proof}
	
	To prove Claims~\ref{claim:2} and \ref{claim:3}, we first observe the following.
	\begin{observation}\label{observation:0} For each $x_{\alpha,i}$ one has $d_\delta(p_\delta(x_{\alpha,i}),v_{\alpha,i})\leq 2\delta$.
	\end{observation}
	\begin{proof}
		Let $p_{\delta}(x_{\alpha,i})=v_{\alpha',i'}$. This means that $x_{\alpha,i}\in V_{\alpha,i}\cap V_{\alpha',i'}$. Therefore, $(v_{\alpha,i},v_{\alpha',i'})\in E_\delta$ is an edge. Since $V_{\alpha,i}$ and $V_{\alpha',i'}$ intersects, so does $U_\alpha$ and $U_{\alpha'}$. Therefore, $d_Z(f_\delta(v_{\alpha,i}),f_\delta(v_{\alpha',i'}))=d_Z(z_\alpha,z_{\alpha'})\leq 2\delta$ establishing that $\mathrm{diam}_Z\big(\{f_\delta(v_{\alpha,i}),f_\delta(v_{\alpha',i'})\}\big)\leq 2\delta$.
	\end{proof}
	
	\begin{proof}[Proof of Claim~\ref{claim:2}]
		Again, we prove only the upper bound since the lower bound proof is similar.
		We have (by Claim~\ref{claim:1})
		$
		\tilde d_f(q(x),q(x_{\alpha,i}))\leq d_\delta(p_\delta(x),p_\delta(x_{\alpha,i}))+\delta
		$
		The righthand side is at most 
		$
		d_\delta(p_\delta(x),v_{\alpha,i})+d_\delta(p_\delta(x_{\alpha,i}),v_{\alpha,i})+\delta
		$
		by triangular inequality. Applying Observation~\ref{observation:0}, we get 
		$
		d_\delta(p_\delta(x),v_{\alpha,i})+d_\delta(p_\delta(x_{\alpha,i}),v_{\alpha,i})+2\delta\leq d_\delta(p_\delta(x),v_{\alpha,i})+ 3\delta
		$
		proving the claim.
	\end{proof}
	
	\begin{proof}[Proof of Claim~\ref{claim:3}]
		We have
		\begin{eqnarray*}
			\tilde d_f(q(x_{\alpha,i}),q(x_{\alpha',i'}) & \leq & d_\delta(p_\delta(x_{\alpha,i}),p_\delta(x_{\alpha',i'}))+\delta\\
			&\leq& d_\delta(p_\delta(x_{\alpha,i}),v_{\alpha,i})+d_\delta(v_{\alpha,i},p_\delta(x_{\alpha',i'}))+\delta\\
			&\leq& d_\delta(v_{\alpha,i},p_\delta(x_{\alpha',i'}))+3\delta\\
			&\leq& d_\delta(v_{\alpha,i},v_{\alpha',i'}) + 5\delta
		\end{eqnarray*}
		The lower bound can be shown similarly.
	\end{proof}
	
\fi

\ifpaper{\subsection{Interleaving of persistent homology groups}} \fi
%%by yusu
\newcommand{\Us}			{{\overline{U}}}
\newcommand{\Usfrak} 		{\mathfrak{\Us}}
\newcommand{\csgood}		{{($c,s$)-good\xspace}}
\newcommand{\DICS}		{\mathfrak{C}}
\newcommand{\Cs}			{\overline{\mathrm C}}
\newcommand{\CC}			{{\mathrm{C}}}
\newcommand{\Csfrak}			{{\overline{\mathfrak{C}}}}
\newcommand{\mycechconst}		{2c}

\ifpaper{\subsubsection{Intrinsic \v{C}ech complex filtrations for $(N({\mathcal U}),d_\delta)$ and for $(X,d_f)$}\label{subsec:mapper}}
\else{\subsection{Intrinsic \v{C}ech complex filtrations for $(N(f^\ast{\mathcal U}),d_\delta)$ and $(X,d_f)$}\label{subsec:mapper}}\fi

\begin{definition}[Intrinsic \v{C}ech filtration]\label{def:IC-filtration}
	The \emph{intrinsic \v{C}ech filtration of the metric space $(Y,d_Y)$} is  
	$$\myCF(Y) = \{ C^r(Y) \subseteq C^{r'} (Y) \}_{0<r<r'}. $$
	The \emph{intrinsic \v{C}ech filtration at resolution $s$} is defined as 
	$\myCF_s(Y) = \{ C^r(Y) \subseteq C^{r'} (Y) \}_{s\le r<r'}. $
\end{definition}

Whenever $(Y,d_Y)$ is totally bounded, the persistence modules induced by taking homology of this intrinsic \v{C}ech filtration become q-tame \cite{CSO14}. This implies that one may define its persistence diagram  $\myDg~\myCF(Y)$ which provides one way to summarize the topological information of the space $Y$ through the lens of its metric structure $d_Y$. 

% If $Y$ is compact (and thus totally bounded in the sense of~\cite{CSO14}), then the persistence module induced by $\myCF(Y)$ is q-tame \cite{CSO14},
%$C^r(Y) = \{(y_0, \ldots, y_p\} \mid 

We prove that the pseudometric space $(X,d_f)$ is totally bounded. This requires us to show that for any $\eps>0$ there is a finite subset of $P\subseteq X$ so that open balls centered at points in $P$ with radii $\eps$ cover $X$. Recall that we have assumed that $X$ is a compact topological space, that $(Z,d_Z)$ is a metric space, and that $f:X\rightarrow Z$ is a continuous map. Consider a cover $\mathcal U$ of $Z$ where each cover element is a ball of radius most $\eps/2$ around a point in $Z$. Then, the pullback cover $f^\ast {\mathcal U}$ of $X$ has all elements with diameter at most $\eps$ in the metric $d_f$. Since $X$ is compact, a finite sub-cover of $f^\ast {\mathcal U}$ still covers $X$. A finite set $P$ consisting of one arbitrary point in each element of this finite sub-cover is such that the union of $d_f$-balls of radius $\epsilon$ around points in $P$ covers $X$. Since $\eps>0$ was arbitrary, $(X,d_f)$ is totally bounded.

% \facundo{REMOVE:
% For results in this and in the next sections, we need that the pseudometric space $(X,d_f)$ is totally bounded, which allows one to use interleaving to bound distance between persistence diagrams. This requires us to show that there is a finite subset of $P\subseteq X$ so that, for any $\eps>0$, balls centering points in $P$ and radii at most $\eps$ cover $X$, see~\cite{CSO14}.
% Given a continuous map $f:X\rightarrow Z$ to a metric space $(Z,d_z)$, 
% one can show that $(X,d_f)$ is totally bounded if $X$ is compact. Consider a cover $\mathcal U$ of $Z$ where each cover element has diameter at most $\eps$. Then, the pullback cover $f^\ast {\mathcal U}$ of $X$ has all elements with diameter at most $\eps$ in the metric $d_f$. Since $X$ is compact, we have a finite sub-cover of $f^\ast {\mathcal U}$. The set of points $P$ consisting of one arbitrary point in each element of this finite sub-cover serves as a witness that $(X,d_f)$ is totally bounded.
%  }

Consider the mapper $N(f^\ast{\mathcal U})$ w.r.t a cover ${\mathcal U}$ of the codomain $Z$. 
We can equip its vertex set, denoted by $\NV$, with a metric structure $(\NV, d_\delta)$, where $\delta$ is  an upper bound on the diameter of each element in $\mathcal U$. 
%${\mathcal U}={\mathcal U}_\delta$, $\delta$ being an upper bound on the diameter of each element in $\mathcal U$. 
Hence we can view the persistence diagram $\myDg~\myCF(\NV)$ w.r.t. the metric $d_\delta$ as a summary of the mapper $N(f^\ast \mathcal U)$. 
Using the Gromov-Hausdorff distance between the metric spaces $(\NV, d_\delta)$ and $(X, d_f)$, we relate this persistent summary to the persistence diagram $\myDg~\myCF(X)$ induced by the intrinsic \v{C}ech filtration of $(X, d_f)$. 
Specifically, we show that $d_{GH} ( (\NV, d_\delta), (X, d_f) ) \le \myc \delta$. 
\ifpaper {} \else Theorem~32 in the full version extends to this result. \fi
With $(X, d_f)$ being totally bounded, by results of \cite{CSO14}\ifpaper{\footnote{Although $d_f$ is a pseudo-metric, the bound on Gromov-Hausdorff distance still implies that the two intrinsic \v{C}ech filtrations $\myCF(\NV)$ and $\myCF(X)$ are interleaved. Now, since $(X,d_f)$ is totally bounded, we can apply results of \cite{CSO14} in our setting.}\else{}\fi, it follows that
	% the two intrinsic \v{C}ech filtrations $\myCF(\NV)$ and $\myCF(X)$ are $\myc %\delta$-interleaved, meaning that 
	the bottleneck-distance between the two resulting persistence diagrams satisfies: 
	\begin{align}\label{eqn:ICbound-mapper}
	d_B(\myDg~\myCF(\NV), \myDg~\myCF(X)) \le 2*\myc\delta = 10 \delta.
	\end{align}
	
	%Hence we can define the intrinsic \v{C}ech complex at scale $r$ to be
	%$C^r(\NV) = \{ (v_0, v_p) \mid v_i \in \NV, \bigcap_{i} B_r(v_i) \neq \emptyset \}$, which is the nerve complex of the set of intrinsic $r$-balls $\{B_r (v) \mid v \in \NV \}$. 
	%Consider the intrinsic-\v{C}ech filtration 
	%$$\myCF(\NV) = \{ C^{r} (\NV) \subseteq C^{r'} (\NV) \}_{0< r < r'}. $$
	%The persistence diagram $\myDg ~\myCF$ of the persistence module induced from this intrinsic-\v{C}ech filtration provides one way to summarize the topological information of the mapper structure $N(X, {\mathcal U})$ through the lens of the metric $d_g$. This persistence summary of $N(X, \mathcal U)$ can be easily related to the persistence diagram $\myDg~\myCF(X)$ induced by the intrinsic-\v{C}ech filtration of $(X, d_f)$, $\myCF(X) = \{ C^r(X) \subseteq C^{r'}(X) \}_{0<r< r'}$, where $C^r(X)$ is the nerve complex of the set of intrinsic $r$-balls $\{B_r( x) = \{ y\in X \mid d_f(x,y) \le r\}, x\in X\}$. 

	\ifpaper{\subsubsection{$\mathrm{MM}(\mathfrak W, f)$ for a tower of covers $\mathfrak W$}\label{subsec:multimapper} }
	\else{\subsection{$\mathrm{MM}(\mathfrak W, f)$ for a tower of covers $\mathfrak W$}\label{subsec:multimapper}}\fi

	Above we discussed the information encoded in a certain persistence diagram summary of a single Mapper structure. 
	We now consider the persistent homology of multiscale mappers. 
	Given any tower of covers (TOC) $\mathfrak W$ of the co-domain $Z$, by applying the homology functor to its multiscale mapper $\mathrm{MM}(\mathfrak{W}, f)$, we obtain a persistent module, and we can thus discuss the persistent homology induced by a tower of covers $\mathfrak W$. 
	However,  as discussed in \cite{DMW16}, this persistent module is not necessarily stable under perturbations (of e.g the map $f$) for general TOCs. 
	To address this issue, Dey et al. introduced a special family of the so-called (c,s)-good TOC in \cite{DMW16}, which is natural and still general. Below we provide an equivalent definition of the (c,s)-good TOC based on the Lebesgue number of covers. 
	
	\begin{definition}[($c,s$)-good TOC] 
		Give a tower of covers $\mathfrak{U} = \{\mathcal U_\eps\}_{\eps \ge s}$, we say that it is \emph{(c,s)-good TOC} if for any $\eps \ge s$, we have that (i) $s_{max}(\mathcal{U}_\eps) \le \eps$ and (ii) $\lambda(\mathcal U_{c\eps}) \ge \eps$. 
	\end{definition}
	
	As an example, the TOC $\mathfrak{U} = \{ \mathcal U_\eps\}_{\eps \ge s}$ with $\mathcal{U}_\eps:= \{ B_{\eps/2}(z) \mid z\in Z\}$ is an (2,s)-good TOC of the co-domain $Z$. 
	
	We now characterize the persistent homology of multiscale mappers induced by (c,s)-good TOCs. 
	Connecting these persistence modules is achieved via the interleaving of towers of simplicial complexes originally introduced in \cite{CCGGO09}. Below we include the slightly generalized version of the definition from \cite{DMW16}. 
	\begin{definition}[Interleaving of simplicial towers, \cite{DMW16}] \label{def:inter-cover}
		Let $\mathfrak{S}=\big\{\mathcal{S}_{\varepsilon}\overset{\tiny{s_{\varepsilon,\varepsilon'}}}{\longrightarrow}\mathcal{S}_{\varepsilon'}\big\}_{r\leq \varepsilon\leq\varepsilon'}$ and $\mathfrak{T}=\big\{\mathcal{T}_{\varepsilon}\overset{\tiny{t_{\varepsilon,\varepsilon'}}}{\longrightarrow}\mathcal{T}_{\varepsilon'}\big\}_{r\leq\varepsilon\leq\varepsilon'}$ be two towers of simplicial complexes where
		$\res(\mathfrak{S})=\res(\mathfrak{T})=r$. For some $c \ge 0$, we say that they are \emph{$c$-interleaved} if for each $\varepsilon\geq r$ one can find simplicial maps $\varphi_\varepsilon:\mathcal{S}_\varepsilon \rightarrow \mathcal{T}_{\varepsilon+c}$ and  $\psi_\varepsilon:\mathcal{T}_\varepsilon \rightarrow \mathcal{S}_{\varepsilon+c}$ so that:
		\begin{itemize} \denselist
			\item[(i)]  for all $\varepsilon\geq r$, $\psi_{\varepsilon+c}\circ\varphi_{\varepsilon}$ and $s_{\varepsilon,\varepsilon+2c}$ are contiguous,
			\item[(ii)]  for all $\varepsilon\geq r$, $\varphi_{\varepsilon+\eta}\circ\psi_{\varepsilon}$ and  $t_{\varepsilon,\varepsilon+2c}$ are contiguous,
			\item[(iii)]  for all $\varepsilon'\geq\varepsilon\geq r$, $\varphi_{\varepsilon'}\circ s_{\varepsilon,\varepsilon'}$ and  $t_{\varepsilon+c,\varepsilon'+c}\circ \varphi_{\varepsilon}$ are contiguous,
			\item[(iv)]  for all $\varepsilon'\geq\varepsilon\geq r$, $ s_{\varepsilon+c,\varepsilon'+c}\circ \psi_{\varepsilon}$ and  $\psi_{\varepsilon'}\circ t_{\varepsilon,\varepsilon'}$ are contiguous.
		\end{itemize}
		Analogously, if we replace the operator `+' by the multiplication `$\cdot$' in the above definition, then we say that $\mathfrak S$ and $\mathfrak T$ are \emph{$c$-multiplicatively interleaved}.  
		
	\end{definition}
	\ifpaper{
		Furthermore, all the simplicial towers that we will encounter here will be those induced by taking the nerve of some tower of covers (TOCs). It turns out that the interleaving of such tower of nerve complexes can be identified via interleaving of their corresponding tower of covers, which is much easier to verify. 
		More precisely,
		%We now introduce the concept of interleaving TOCs used in \cite{DMW16}, which is much simpler to verify than interleaving simplicial towers.  
		%It turns out \cite{DMW16} that interleaving of TOCs is much simpler to identify and interleaving TOCs lead to interleaving between the towers of their nerve complexes. More specifically, 
		
		\begin{definition}[(Multiplicative) Interleaving of towers of covers, \cite{DMW16}] Let $\mathfrak{V} = \{\mathcal{V}_\varepsilon\}$ and $\mathfrak{W}=\{\mathcal{W}_\varepsilon\}$ be two towers of covers of a topological space $X$ such that $\res(\mathfrak{V})=\res(\mathfrak{W})=r$. Given $\eta\geq 0$,  we say that $\mathfrak{V}$ and $\mathfrak{W}$ are 
			\emph{$\eta$-multiplicatively interleaved} if one can find maps of covers $\zeta_\varepsilon:\mathcal{V}_\varepsilon \rightarrow \mathcal{W}_{\eta \cdot \varepsilon}$ and $\xi_{\varepsilon'}:\mathcal{W}_{\varepsilon'} \rightarrow \mathcal{V}_{\eta \cdot \varepsilon'}$ for all $\varepsilon,\varepsilon'\geq r.$ 
		\end{definition}
		
		The following two results of \cite{DMW16} connect interleaving TOCs with the interleaving of their induced tower of nerve complexes and multiscale mappers \footnote{These propositions are proven in \cite{DMW16} for the additive version of interleaving; but the same proofs hold for the multiplicative interleaving case. }. 
		
		\begin{proposition}[Proposition 4.2 of \cite{DMW16}]\label{prop:interleaving}
			Let $\mathfrak{U}$ and $\mathfrak{V}$ be two $\eta$-(multiplicatively) interleaved towers of covers of $X$ with $\res(\mathfrak{U})=\res(\mathfrak{V})$. Then, $N(\mathfrak{U})$ and $N(\mathfrak{V})$ are also $\eta$-(multiplicatively) interleaved.
		\end{proposition}
		
		\begin{proposition}[Proposition 4.1 of \cite{DMW16}]\label{prop:inter-pullback}
			%(i) If  $\mathfrak{U}$ and $\mathfrak{V}$ are $\eta_1$-interleaved and $\mathfrak{V}$ and $\mathfrak{W}$ are $\eta_2$-interleaved, then, $\mathfrak{U}$ and $\mathfrak{W}$ are $(\eta_1+\eta_2)$-interleaved. 
			%(ii) 
			Let $f:X\rightarrow Z$ be a continuous function and $\mathfrak{U}$ and 
			$\mathfrak{V}$ be two $\eta$-(multiplicatively) interleaved tower of covers 
			of $Z$. Then, $f^\ast(\mathfrak{U})$ and 
			$f^\ast(\mathfrak{V})$ are also $\eta$-(multiplicatively) interleaved. 
			
			By Proposition \ref{prop:interleaving}, this implies that the resulting multiscale mappers $\mathrm{MM}(\mathfrak{U}, f)$ and $\mathrm{MM}(\mathfrak{U}, f)$ are also $\eta$-(multiplicatively) interleaved. 
		\end{proposition}
	}\fi
	
	Our main results of this section are the following \ifpaper{. }\else{whose proofs are deferred to the full version. }\fi
	First, Theorem \ref{thm:goodTOCinterleave} states that the multiscale-mappers induced by any two ($c,s$)-good towers of covers interleave with each other, implying that their respective persistence diagrams are also close under the bottleneck distance. From this point of view, the persistence diagrams induced by any two (c,s)-good TOCs contain roughly the same information. 
	Next in Theorem \ref{thm:MM-ICinterleave}, we show that the multiscale mapper induced by any $(c,s)$-good TOC interleaves (at the homology level) with the intrinsic \v{C}ech filtration of $(X,d_f)$, thereby implying that the persistence diagram of the multiscale mapper w.r.t. any ($c,s$)-good TOC is close to that of the intrinsic \v{C}ech filtration of $(X, d_f)$ under the bottleneck distance. 
	
	\begin{theorem}\label{thm:goodTOCinterleave}
		Given a map $f: X \to Z$, let 
		$\mathfrak{V}=\{\mathcal{V}_{\varepsilon}\overset{\tiny{v_{\varepsilon,\varepsilon'}}}{\longrightarrow} \mathcal{V}_{\varepsilon'}\big\}_{\varepsilon\leq \varepsilon'}$ and 
		$\mathfrak{W}=\{\mathcal{W}_{\varepsilon}\overset{\tiny{w_{\varepsilon,\varepsilon'}}}{\longrightarrow} \mathcal{W}_{\varepsilon'}\big\}_{\varepsilon\leq \varepsilon'}$ 
		be two ($c, s$)-good tower of covers of $Z$. 
		Then the corresponding multiscale mappers $\mathrm{MM}(\mathfrak{V},f)$ and $\mathrm{MM}(\mathfrak W, f)$ are $c$-multiplicatively interleaved. 
	\end{theorem}
	\ifpaper{
		\begin{proof}
			First, we make the following observation. 
			\begin{claim}
				Any two ($c,s$)-good TOCs $\mathfrak{V}$ and $\mathfrak{W}$ are $c$-multiplicatively interleaved. 
			\end{claim}
			\begin{proof}
				It follows easily from the definitions of $(c,s)$-good TOC. 
				Specifically, first we construct $\zeta_\varepsilon:\mathcal{V}_\varepsilon \rightarrow \mathcal{W}_{c \cdot \varepsilon}$. 
				For any $V \in \mathcal{V}_\eps$, we have that $\diam(V) \le \eps$. Furthermore, since $\mathfrak{W}$ is $(c,s)$-good, there exists $W \in \mathcal{W}_{c \eps}$ such that $V \subseteq W$. Set $\zeta_\eps(V) = W$; if there are multiple choice of $W$, we can choose an arbitrary one. 
				We can construct $\xi_{\varepsilon'}:\mathcal{W}_{\varepsilon'} \rightarrow \mathcal{V}_{c \cdot \varepsilon'}$ in a symmetric manner, and the claim then follows. 
			\end{proof}
			This, combined with Propositions \ref{prop:inter-pullback} and \ref{prop:interleaving}, prove the theorem. 
		\end{proof}
		
		Recall the definition of intrinsic \v{C}ech complex filtration $\myCF-s(Y)$ at resolution $s$ for a metric space $(Y,d_Y)$ in Def \ref{def:IC-filtration}.
	}\fi
	
	\begin{theorem}\label{thm:MM-ICinterleave}
		%Let $\specialU_s$ be as defined in Eqn (\ref{eqn:specialU} and $\mathfrak{C}_s(X)$ the intrinsic \v{C}ech filtration of $(X, d_f)$ starting with resolution $s$. 
		Let $\mathfrak{C}_s(X)$ be the intrinsic \v{C}ech filtration of $(X, d_f)$ starting with resolution $s$. Let $\mathfrak{U}=\{\mathcal{U}_{\varepsilon}\overset{\tiny{u_{\varepsilon,\varepsilon'}}}{\longrightarrow} \mathcal{U}_{\varepsilon'}\big\}_{s\le \varepsilon\leq \varepsilon'}$ be a ($c,s$)-good TOC of the compact connected metric space $Z$. 
		Then the multiscale mapper $\mathrm{MM} (\mathfrak U, f)$ and $\mathfrak{C}_s(X)$ are $2c$-multiplicatively interleaved. % at the homology level. 
	\end{theorem}
	\ifpaper
	\begin{proof}
		Let $D_\eps := \{ B_\eps(x) \mid x\in X\}$ be the infinite cover of $X$ consisting of all $\eps$-intrinsic balls in $X$. 
		Obviously, $C^\eps(X)$ is the nerve complex induced by the cover $D_\eps$, for each $\eps > 0$. 
		Let $\mathfrak{D} = \{ D_\eps \overset{\tiny{t_ {\varepsilon,\varepsilon'}}}{\longrightarrow} D_{\eps'} \}_{s\le \eps < \eps'}$ be the corresponding tower of covers of $X$, where $t_{\eps, \eps'}$ sends $B_\eps(x) \in D_\eps$ to $B_{\eps'}(x) \in D_{\eps'}$. 
		Obviously, the tower of the nerve complexes for $D_\eps$s give rise to $\mathfrak{C}_s(X)$. 
		
		On the other hand, let $W_\eps = f^*\mathcal{U}_\eps$ be the pull-back cover of $X$ induced by $\mathcal{U}_\eps$ via $f$, and $\mathfrak{W} = f^* \mathfrak{U}$ is the pull-back tower of cover of $X$ induced by the TOC $\mathfrak{U}$ of $Z$. 
		By definition, we know that the multiscale mapper $\mathrm{MM}(\mathfrak U, f) = \{\M_\eps \overset{\tiny{s_{\varepsilon,\varepsilon'}}}{\longrightarrow} \M_{\eps'}\}_{s\le \eps \le \eps'}$ where $\M_\eps$ is the nerve complex of the cover $W_\eps$. 
		
		In what follows, we will argue that the two TOCs $\mathfrak{D}$ and $\mathfrak{W}$ are $2c$-multiplicatively interleaved. By Proposition \ref{prop:interleaving}, this then proves the theorem. 
		
		First, we show that there is a  map of covers $\zeta_\eps: D_\eps \to W_{2c\eps}$ for each $\eps \ge s$ defined as follows. 
		
		Take any intrinsic ball $B_{\eps, d_f} (x) \in D_\eps$ for some $x\in X$. 
		Consider the image $f(B_\eps(x)) \subseteq Z$. Recall that the covering metric $d_f(x_1, x_2)$ on $X$ is defined by  the minimum diameter of the image of any path $\rho$ connecting $x_1$ to $x_2$ in $X$; that is, $d_f(x_1, x_2) = \inf_{\rho: x\leadsto y} \diam(f(\rho))$. Thus $d_Z(f(x_1), f(x_2)) \le d_f(x_1, x_2)$. 
		We then have that for any $x_1, x_2 \in B_\eps(x)$, 
		$$d_Z(f(x_1), f(x_2)) \le d_Z(f(x_1), f(x)) + d_Z(f(x), f(x_2)) \le d_f(x_1, x) + d_f(x_2, x) \le 2\eps.$$
		This implies that $\diam(f(B_\eps(x))) \le 2\eps$. 
		Since $\mathfrak U$ is a $(c,s)$-good TOC, it then follows that there exists $U_x \in \mathcal{U}_{2c\eps}$ such that $f(B_\eps(x)) \subseteq U$ (if there are multiple elements contains $f(B_\eps(x))$, we can choose an arbitrary one as $U_x$). 
		This means that $B_\eps(x)$ is contained within one of the connected component, say $W_x$ in $\mathrm{cc}(f^{-1}(U_x))$. We simply set $\zeta_\eps( B_\eps(x) ) = W_x \in W_{2c\eps}$. 
		
		Finally we show that there is a map of covers $\xi_\eps: W_\eps \to D_\eps (\overset{\tiny{t_{\varepsilon,2\varepsilon}}}{\longrightarrow} D_{2c\eps})$. 
		To this end, consider any set $V \in W_\eps$; by definition, there exists some $U \in \mathcal{U}_\eps$ such that $V \in \mathrm{cc}(f^*(U))$.
		Note, $f(V) \subseteq U$ and $\diam(U) \le s_{max}(\mathcal{U}_\eps) \le \eps$. 
		It then follows from the definition of the metric $d_f$ that for any point $x$ from $V$, we have that $V \subseteq B_{\eps, d_f}(x)$. 
		We simply set $\xi_\eps(V) = B_\eps(x)$. 
		This completes the proof that the two TOCs $\mathfrak{D}$ and $\mathfrak{W}$ are $2c$-multiplicatively interleaved. The theorem then follows this and Proposition \ref{prop:interleaving}. 
	\end{proof}
	%\else
	%Proof is deferred to the full version.
	\fi
	Finally, given a persistence diagram $\myDg$, we denote its \emph{log-scaled version} $\myDg_{\log}$ to be the diagram consisting of the set of points $\{ (\log x, \log y) \mid (x,y) \in \myDg \}$. Since interleaving towers of simplicial complexes induce  interleaving persistent modules, using results of \cite{CCGGO09,CSGO16}, we have the following corollary. 
	\begin{corollary}\label{cor:MM-IC-Dgdistance}
		Given a continuous map $f: X \to Z$ and a ($c,s$)-good TOC $\mathfrak{U}$ of $Z$, let $\myDg_{\log} \MM(\mathfrak{U}, f)$ and $\myDg_{\log} \mathfrak{C}_s$ denote the log-scaled persistence diagram of the persistence modules induced by $\MM(\mathfrak{U}, f)$ and by the intrinsic \v{C}ech filtration $\mathfrak{C}_s$ of $(X, d_f)$ respectively. 
		We have that 
		$$d_B(\myDg_{\log} \MM(\mathfrak{U}, f), \myDg_{\log} \mathfrak{C}_s) \le 2c . $$
	\end{corollary}
	
	\ifpaper{
		\section{Concluding remarks}
		In this paper, we present some studies on the topological information encoded in Nerves, Reeb spaces, mappers and multiscale mappers , where the latter two structures are constructed based on nerves. 
		Currently, the characterization for the $H_1$-homology for the Nerve complex is much stronger than for higher dimensions. In particular, we showed that for a path-connected cover $\mathcal{U}$, there is a surjection from the domain $H_1(X)$ to $H_1(N(\mathcal{U}))$. 
		While this does not hold for higher dimensional cases (as Figure \ref{non-surject-fig} demonstrates), we wonder if similar surjection holds under additional conditions on the input cover such as the ones used by Bj\"{o}rner \cite{B03} for homotopy groups. Along that line, we ask: if for any $k\geq 0$, $t$-wise intersections of cover elements for all $t>0$ have trivial reduced homology groups for all dimensions up to $k-t$, then does the nerve map induce a surjection for the $k$-dimensional homology? We have answered it affirmatively for $k=1$.
		
		We also remark that it is possible to carry out most of our arguments using the language of category theory (see e.g, \cite{Stovner12} on this view for the mapper structure). We choose not to take this route and explain the results with more elementary expositions. 
	}\else{}\fi
	
	\myparagraph{Acknowledgments.} 
	We thank the reviewers for helpful comments. This work was partially supported by National Science Foundation under grant CCF-1526513.
	
	%\subparagraph*{Acknowledgments}
	%
	%I want to thank \dots
	
	%\appendix

	%%
	%% Bibliography
	%%
	
	%\bibliographystyle{plain}
	%\bibliography{MMprop}
	%% Either use bibtex (recommended), but commented out in this sample
	
	%\bibliography{dummybib}
	
	%% .. or use bibitems explicitly
	
	%\nocite{Simpson}

	%\appendix
	
\end{document}